\documentclass[twoside,11pt]{article}
\usepackage{amsmath,amssymb,amsthm,amsfonts,mathrsfs,wasysym,latexsym,times,lineno, subfigure,color,booktabs}
\usepackage{amsmath,amsfonts}
\usepackage{multirow}
\usepackage{array}
\usepackage{amsthm}
\usepackage{graphicx}
\usepackage{color}
\usepackage{multicol}
\usepackage{float}

\newtheorem{thm}{Theorem}

\newtheorem{lemma}[thm]{Lemma}

\newtheorem{theorem}[thm]{Theorem}

\newtheorem{proposition}[thm]{Proposition}

\newtheorem{corollary}[thm]{Corollary}

\allowdisplaybreaks[4]
  \date{}

\title{Hiding solutions in model RB: Forced instances are almost as hard as unforced ones}

\author{Guangyan Zhou$^*$\\

\footnotesize Department of Mathematics, Beijing Technology and Business University, Beijing 100048, China}

\begin{document}
\maketitle
\begin{abstract}

In this paper we study the forced instance spaces of model RB, where one or two arbitrary satisfying assignments have been imposed. We prove rigorously that  the expected number  of solutions of forced RB instances is asymptotically the same with those of unforced ones. Moreover, the distribution of forced RB instances in the corresponding forced instance space is asymptotically the same with that of unforced RB instances in the unforced instance space. These results imply that the hidden assignments will not lead to easily solvable formulas, and the hardness of solving forced RB instances will be the same with unforced RB instances.

\end{abstract}

\footnotetext[1]{zhouguangyan@btbu.edu.cn. Project supported by the National Natural Science Foundation of China (No. 61702019).}

\section{Introduction}
In mathematics, computer science and statistical physics, there exist many constraint satisfaction problems (CSPs) whose worst-case computation is of exponential complexity. Studies show that solution space structures affect the performance of algorithms and the hardness of problems, and the computing complexity attains the maximum at the satisfiability threshold, where the probability of a random instance being satisfiable changes from 1 to 0 sharply. These hard problems provide challenging benchmarks for evaluation of algorithms. To evaluate the performance of incomplete algorithms around the hard region, it is natural to generate hard but satisfiable benchmarks, which can be done by planting solutions to random CSPs \cite{bar02,Dimitris05,jia2007,krza2009,xu2007}. In graph theory, planted cliques problems \cite{alon98,feige00,feige10,dekel14,deshpande15}, planed partitions \cite{amin09} and planted colourings \cite{bapst17} have been studied  as well. Generating hard satisfiable instances has important applications, e.g. one-way function in cryptography. Indeed, \cite{ju1998} applied
the ``planting" strategy  to random graph problems, and proved that large hidden cliques are as hard to find as large cliques in usual random graphs, thus hidden cliques can be applied to create cryptographically secure primitives.

Sometimes planting a solution may change the properties of the ensemble, thus making it difficult to understand when and why they can provide hard instances. For instance,  random $k$-SAT instances generated by planting an assignment tend to be easy, where the hidden assignment attracts assignments close to it, thus can be solved efficiently. However, ``2-hidden" (where two complementary assignments are hidden) \cite{Dimitris05} and ``$q$-hidden" \cite{jia2007} 3-SAT formulas are as hard as usual random 3-SAT formulas. In this paper we will study the model RB, in which the properties of the usual random instances and random instances with one hidden solution are identical in that, the expected number of solutions and the hardness at the threshold remain the same \cite{xu2000}.

Model RB is a standard prototype CSP model with growing domains revised from the Model B \cite{smith1996}. The proposal of this model is based on the fact that, instances of standard CSP models A, B, C and D do not have an asymptotic threshold due to the presence of flawed variables. To eliminate such flaws, alternative models have been proposed either by incorporating certain combinatorial structure on constraints to ensure certain consistency properties \cite{gao2004,gent2001}, or enlarging the scales of parameters such as domain size and constraint length \cite{smith2001,frieze2006,fan2011,fan2012}. Model RB has been proved to have sharp SAT-UNSAT phase transition and exact  threshold points, and can generate hard instances in the phase transition region \cite{xu19,xu2006}. Solution space structure of model RB was studied in \cite{zhao2012} and \cite{xuwei2015}, and it was proved that clustering phase exists where solutions are clustered into an exponential number of well-separated clusters and each cluster contains a sub-exponential number of solutions, and persists until the satisfiability transition. Thus Model RB has a clustering transition but no condensation transition, and this is different to other classic random CSPs such as random $k$-SAT and graph colouring problems which has a condensation transition.

Further, \cite{xu2007} considered forced satisfiable RB instances on which a solution is imposed. Theoretical and experimental results show that such strategy has almost no effect on the number of solutions and does not lead to a biased sampling of instances with many solutions. The reason is that the distribution of the solutions of model RB is uniform, while conversely for random $k$-SAT, the distribution is highly skewed in that satisfying assignments tend to form clumps, thus causing a biased sampling of instances with many solutions clustered around the hidden solution \cite{Dimitris05,xu2006}. In addition, the hardness of forced satisfiable instances is similar to unforced satisfiable ones, which is supported by theoretical and experimental results. Benchmarks based on model RB have been widely used to test the performance of different algorithms. Experimental results show that the hardness of solving both forced and unforced RB instances grows exponentially with $n$ ($n$ is the number of variables), and instances with $n=59$ appear to be challenging to solve, thus confirm the hardness of benchmarks. Another interesting aspect is that, based on model RB, a random graph model can be proposed in which an  exact maximum independent set can be hidden, since determining whether a trivial upper bound of the size of maximum independent sets can be reached is equivalent to determining the satisfiability of the corresponding instances of the model RB.

In this paper, we hide two arbitrary solutions in random RB instances. We rigorously show that the expected number of solutions of a random forced RB instance with two hidden satisfying assignments, is asymptotically the same with that of unforced RB instance space. Indeed for a random RB instance, the third moment of the number of satisfying assignments equals asymptotically  to the cube of the first moment, implying that the solution space structure of model RB is uniformly distributed, and that the hidden assignments are not likely to attract more satisfying assignments, thus will not lead to easily solvable formulas. More importantly, the distribution of forced RB instances with either one or two hidden solutions in the corresponding instance space do not change compared with unforced ones. Therefore the hardness of solving forced RB instances will be exactly the same with unforced ones, while for other problems, the running time of solving the usual random ensemble and planted random ensemble will differ by a polynomial term. Our analyses are based on a third moment argument which leads to the estimation of a sum with several parameters. The sum is dominated by a small interval which includes a global maximum point, while the  contributions of remaining intervals are negligible. Moreover, from the process of our analysis we can see that, the two hidden assignments have little effect on the dominate assignments which contribute the most to the expected number of solutions.

This paper is organized as follows. In Section \ref{result} we introduce the model RB and present our main results. In Section \ref{proof1} we prove Theorem \ref{expectation}, which is the foundation of proving other main results. The proofs of the remaining results are in Section \ref{proof:th2}. Finally, we present possible application to cryptography, and discuss some future work.

\section{Model RB and Main results}\label{result}
\subsection{Model RB}

 A random instance $I$ of the model RB consists of the following:
 \begin{itemize}
   \item  A set of variables $X=\{x_1,...,x_n\}$: for each $i=1,...,n$,
   variable $x_i$ takes values from its domain $D_i$ and all with the same domain size $|D_i|=d$ ($d=n^\alpha$, where $\alpha>0$ is a constant).
   \item  A set of constraints $C=\{C_1,...,C_m\}$ ($m=rn\ln n$, where $r>0$ is a constant): for each
   $i=1,...,m$, constraint $C_i=(X_i,R_i)$. $X_i=(x_{i_1},x_{i_2},...,x_{i_k})$
   ($k\geq2$ is a constant) is a sequence of $k$ distinct variables chosen
   uniformly at random without repetition from $X$. $R_i$
   is the permitted set of tuples of values which are selected uniformly without repetition
   from the subsets of $ D_{i_1}\times
D_{i_2}\times\cdots\times D_{i_k}$, and $|R_i|=(1-p)d^k$ where $0<p<1$
is a constant.
 \end{itemize}

Xu and Li \cite{xu2000} proved that model RB has a satisfiable-unsatisfiable transition: Assume that $\alpha>\frac1k,0<p<1,k\ge2$ are constants, and $k\ge\frac1{1-p}$. Let $r_c=-\frac{\alpha}{\ln(1-p)}$, then
\begin{eqnarray*}
\nonumber \lim_{n\rightarrow\infty}\mbox{\textbf{Pr\ }}[I\ \mbox{is
satisfiable}]=\left\{
                                                        \begin{array}{cc}
                                                          1 & \mbox{\quad$r$ $<$ $r_c$}, \\
                                                          0 & \mbox{\quad$r$ $>$ $r_c$}. \\
                                                        \end{array}
                                                      \right.
\end{eqnarray*}
In the following, we tacitly assume that $r<r_c$ and $\alpha,p,k$ are constants satisfying the conditions above.

\emph{\textbf{Generate random RB instances with hidden solutions.}}
The method of generating forced satisfiable instances of model RB with one (or two) hidden solutions follows from \cite{xu2006}:
 \begin{itemize}
   \item  Generate a random solution $\sigma$ (or two solutions $\sigma,\tau$).
   \item  Generate $m$ constraints of size $k$ (with repetition).
\item For each constraint, generate $nb$ allowed tuples, where $nb$ is the nearest integer to $(1-p)d^k$.
\begin{itemize}
\item Select as first allowed tuple which satisfies the solution $\sigma$ (resp. $\sigma,\tau$).
\item Select $nb-1$ other tuples (without repetition).
\end{itemize}
 \end{itemize}

Theoretical and experimental results in \cite{xu2006} show that, the expected number of solutions of random RB instances in which one solution is hidden remains the same, and the complexity of solving such random forced  RB instances with one hidden solution near the threshold has exponential lower bounds. In the following, we present our main results based on the study of the solution space structures and distribution of forced RB instances with one or two hidden solutions.

\subsection{Main results}

Let $\mathbf{E}[N]$ be the expected number of solutions over all possible RB instances, and  $\mathbf{E}_f[N^{**}]$ be the expected number of solutions over all forced RB instances on which two solutions are imposed. Our first main result asserts that $\mathbf{E}[N]$ and $\mathbf{E}_f[N^{**}]$ are asymptotically equal within the satisfiable phase.

\begin{theorem}\label{expectation}
Suppose that $r<r_c$, then
\begin{align*}
\lim_{n\rightarrow\infty}\frac{\mathbf{E}_f[N^{**}]}{\mathbf{E}[N]}=1.
\end{align*}
\end{theorem}

We proceed to highlight a few interesting consequences of Theorem \ref{expectation}. Before we carry out the results let us specify some notations.
Consider the instance space $\mathbb{S}$ of model RB, and let $\mathcal{P}$ be the probability distribution of the number of solutions in the entire instance space. Precisely, if we let $\mathcal{E}_t$ be the event that a RB instance has $t$ solutions (where $t$ is a positive integer), then $\mathcal{P}_t=\mathcal{P}(\mathcal{E}_t)$ represents the probability of a random instance  containing exactly $t$ solutions. Correspondingly, let $\mathbb{S}^*$ be the space of instances on which a random solution (say $\sigma$) is imposed, and $\mathcal{P}_t^*=\mathcal{P}(\mathcal{E}_t|\sigma\text{ is a solution})$ be the probability distribution of instances with $t$ solutions over $\mathbb{S}^{*}$; let $\mathbb{S}^{**}$ be the space of instances on which two arbitrary solutions (say $\sigma,\tau$) are imposed, and $\mathcal{P}_t^{**}=\mathcal{P}(\mathcal{E}_t|\sigma,\tau\text{ are solutions})$ be the probability distribution of instances with $t$ solutions over $\mathbb{S}^{**}$.

The following result determines the relations of the asymptotic distributions of RB instances with exact number of solutions among the three instance spaces mentioned above.

\begin{corollary}\label{distribution}
Assume that  $r<r_c$ and $t$ is a positive integer, then the distributions of  RB instances with $t$ solutions in instance spaces $\mathbb{S},\mathbb{S}^*$ and $\mathbb{S}^{**}$ are asymptotically the same as $n\rightarrow\infty$, i.e.,
\begin{align*}
\mathcal{P}_t^{**}\sim\mathcal{P}_t^*\sim\mathcal{P}_t.
\end{align*}
\end{corollary}

Theorem \ref{expectation} also yields the following result, which implies that the solutions of random RB instances are distributed uniformly.
\begin{corollary}\label{thm:thirdmoment}
Assume that  $r<r_c$, and $N$ being the number of solutions of a random  unforced RB instance, then
\begin{align*}
\lim_{n\rightarrow\infty}\frac{\mathbf{E}[N^3]}{\mathbf{E}[N]^3}=1.
\end{align*}
\end{corollary}

\begin{theorem}\label{thm:hard}
In the satisfiable phase of Model RB with all parameters being fixed, an algorithm can solve a randomly generated forced RB instance with one or two arbitrary hidden solutions with probability $P$, if and only if it can solve a randomly generated unforced RB instance with probability $P$.
\end{theorem}

\section{Proof of Theorem \ref{expectation}}\label{proof1}
Suppose $\sigma,\tau$ are two fixed assignments with similarity number $S_0$, i.e., $S_0=|\mathcal{S}_0|=|\{x\in X: \sigma(x)=\tau(x)\}|$.
The RB instances we consider in the following are from the space $\mathbb{S}^{**}$, in which instances are generated subject to the condition that both $\sigma$ and $\tau$ are solutions. From \cite{xu2000}, the expected number of solutions for unforced RB instances is
\begin{align*}\label{eqn:en}
\mathbf{E}[N]=d^n(1-p)^{rn\ln n},
\end{align*}
and the expected number of satisfying assignment pairs with similarity number $S_0$ is
\begin{equation*}
A_{S_0}\mathbf{Pr(<\sigma,\tau>)},
\end{equation*}
 where
\begin{equation*}\label{eqn:as0}
A_{S_0}=d^n(d-1)^{n-S_0}\binom n{S_0},\mathbf{Pr}(<\sigma,\tau>)=\left[\frac{\binom{d^k-1}{q}}{\binom{d^k}{q}}\frac{\binom Sk}{\binom nk}+\frac{\binom{d^k-2}{q}}{\binom{d^k}{q}}\left(1-\frac{\binom Sk}{\binom nk}\right)\right]^{rn\ln n}.
\end{equation*}

For a random instance $I\in\mathbb{S}^{**}$, let $\omega$ be a random assignment. Assume that the similarity sets of variables assigned by $\omega$ with $\sigma,\tau$ are:
\begin{align*}
&\mathcal{S}_1=\{x\in X: \sigma(x)=\omega(x)\},S_1=|\mathcal{S}_1|,\\
&\mathcal{S}_2=\{x\in X: \tau(x)=\omega(x)\},S_2=|\mathcal{S}_2|,\\
&\mathcal{S}=\{x\in X: \sigma(x)=\tau(x)=\omega(x)\},S=|\mathcal{S}|.
\end{align*}
To simplify the calculation in the subsequent sections, it is sometimes convenient to consider the sets $\mathcal{M}_1=\mathcal{S}_1/\mathcal{S},\mathcal{M}_2=\mathcal{S}_2/\mathcal{S}$ with $M_1=|\mathcal{M}_1|=S_1-S,M_2=|\mathcal{M}_2|=S_2-S$. Given $\sigma,\tau$, the number of such $\omega$ equals

\begin{equation*}
A_{S_0,S_1,S_2,S}=d^n(d-1)^{n-S_0}(d-2)^{n-(S_1+S_2-S)}\binom{n}{S_0-S,S_1-S,S_2-S,S,n-S_0-S_1-S_2+2S}.
\end{equation*}

For a random constraint $C$,

(1) The probability of $<\sigma,\tau,\omega>$ satisfying a random constraint $C$,
in which each variable is assigned the same values by $\sigma,\tau,\omega$ (we denote by $\sigma(C)=\tau(C)=\omega(C)$), is
\begin{equation*}
\frac{\binom Sk}{\binom nk}\frac{\binom{d^k-1}{q}}{\binom{d^k}{q}}.
\end{equation*}

(2) The probability of $<\sigma,\tau,\omega>$ satisfying a random constraint $C$,
in which each variable is assigned the same values by two of $\sigma,\tau,\omega$, while differently by the third assignment, is
\begin{equation*}
\left(\frac{\binom{S_0}k-\binom Sk}{\binom nk}+\frac{\binom{S_1}k-\binom Sk}{\binom nk}+\frac{\binom{S_2}k-\binom Sk}{\binom nk}\right)\frac{\binom{d^k-2}{q}}{\binom{d^k}{q}}.
\end{equation*}

Indeed, suppose  $\sigma(C)=\tau(C)\ne\omega(C)$), the probability of $C$ falling into such a case is $\frac{\binom{S_0}k-\binom Sk}{\binom nk}$, and the probability of $<\sigma,\tau,\omega>$ satisfying such $C$ is
$\frac{\binom{d^k-2}{q}}{\binom{d^k}{q}}.$
Similar results hold for the cases $\sigma(C)=\omega(C)\ne\tau(C)$ and $\omega(C)=\tau(C)\ne\sigma(C)$.

(3) The probability of $<\sigma,\tau,\omega>$ satisfying a random constraint $C$ which  falls into the remaining cases is
\begin{equation*}
\left(1-\frac{\binom{S_1}k+\binom{S_2}k+\binom{S_0}k-2\binom Sk}{\binom nk}\right)
\frac{\binom{d^k-3}{q}}{\binom{d^k}{q}}.
\end{equation*}
Consequently, the probability of $<\sigma,\tau,\omega>$, which has the above overlaps, satisfying a random constraint $C$ is

\begin{align*}
&\mathbf{Pr}(<\sigma,\tau,\omega>)\\
=&\frac{\binom Sk}{\binom nk}\cdot\frac{\binom{d^k-1}{q}}{\binom{d^k}{q}}+\frac{\binom{S_1}k+\binom{S_2}k+\binom{S_0}k-3\binom Sk}{\binom nk}\cdot\frac{\binom{d^k-2}{q}}{\binom{d^k}{q}}+\left(1-\frac{\binom{S_1}k+\binom{S_2}k+\binom{S_0}k-2\binom Sk}{\binom nk}\right)\frac{\binom{d^k-3}{q}}{\binom{d^k}{q}}.
\end{align*}

The expected number of solutions of a random instance $I\in\mathbb{S}^{**}$ with two hidden assignments $\sigma,\tau$ of similarity number $S_0$ equals
\begin{align*}
\mathbf{E}_{f}[N^{**}]&=\sum_{i=1}^{d^n}\mathbf{Pr}[\omega_i\text{ is a satisfying assignment }|\sigma,\tau\text{ are satisfying assignments}]\\
&=\sum_{i=1}^{d^n}\frac{\mathbf{Pr}[\sigma,\tau,\omega_i\text{ are satisfying assignments}]}{\mathbf{Pr[\sigma,\tau\text{ are satisfying assignments}]}}\\
&=\sum_{(S_1,S_2,S)\in\mathfrak{S}}\frac{A_{S_0,S_1,S_2,S}\mathbf{Pr}(<\sigma,\tau,\omega>)^{rn\ln n}}{A_{S_0}\mathbf{Pr}(<\sigma,\tau>)^{rn\ln n}},
\end{align*}
where $\mathfrak{S}=\{(S_1,S_2,S):0\le S_1,S_2\le n,0\le S\le\min\{S_0,S_1,S_2\},S_1+S_2-2S\le n-S_0\}$, and this set is equivalent to $\mathfrak{M}=\{(M_1,M_2,S):0\le M_1+S,M_2+S\le n,0\le S\le S_0,M_1+M_2\le n-S_0\}$.

The following asymptotic analyses will be needed.
\begin{equation*}
\frac{\binom{d^k-1}{q}}{\binom{d^k}{q}}=1-p,\frac{\binom{d^k-2}{q}}{\binom{d^k}{q}}=(1-p)^2+O(d^{-k}),\frac{\binom{d^k-3}{q}}{\binom{d^k}{q}}=(1-p)^3+O(d^{-k}).
\end{equation*}

Let $S=ns$, then
\begin{equation*}
\frac{\binom Sk}{\binom nk}=s^k+\frac{g(s)}n+O\left(\frac1{n^2}\right),\text{ where }g(s)=\frac{k(k-1)(s^k-s^{k-1})}2.
\end{equation*}

Combining the above, we get
\begin{equation}\label{eqn:main}
\frac{\mathbf{E}_{f}[N^{**}]}{\mathbf{E}[N]}\le\left(1+O\left(\frac1n\right)\right)\sum_{(S_1,S_2,S)\in\mathfrak{S}}\Psi(S_1,S_2,S),
\end{equation}
with $S_1=ns_1,S_2=ns_2,S=ns$, and
\begin{align*}
\Psi(S_1,S_2,S)=&
\frac{\binom{n}{S_0-S,S_1-S,S_2-S,S,n-S_0-S_1-S_2+2S}}{\binom{n}{S_0}}\frac{(d-2)^{n-(S_1+S_2-S)}}{d^n}\left[1+\frac{p(s_1^k+s_2^k)+\frac{2p^2-p}{1-p}s^k}{1-p+ps_0^k}\right]^{rn\ln n}\\
=&\exp\{n\psi(s_1,s_2,s)\},\\
\psi(s_1,s_2,s)=&-(s_1-s)\ln (s_1-s)-(s_2-s)\ln (s_2-s)-(s_0-s)\ln(s_0-s)-s\ln s+s_0\ln s_0\\
&+(1-s_0)\ln(1-s_0)-(1-s_1-s_2-s_0+2s)\ln(1-s_1-s_2-s_0+2s)
-\ln d\\
&+(1-s_1-s_2+s)\ln(d-2)+r\ln n\ln\left[1+\frac{ps_1^k+ps_2^k+(2p^2-p)s^k/(1-p)}{1-p+ps_0^k}\right],
\end{align*}
where $\psi(s_1,s_2,s)$ is asymptotically obtained by applying Stirling's formula for the factorial.

To prove Theorem \ref{expectation}, we consider two cases of $\alpha\ge1$ and $\alpha<1$. The crucial point is to decompose the sum (\ref{eqn:main})  into several parts which are tractable by fairly and delicately considering the sizes of overlap sets of the three assignments, and show that $\frac{\mathbf{E}_{f}[N^{**}]}{\mathbf{E}[N]}\le1+o(1)$. Indeed, by H$\ddot{o}$lder inequality, $\mathbf{E}[N]=\mathbf{E}[N1_{N>0}]\le(\mathbf{E}[N^3])^\frac13\mathbf{P}(N>0)^\frac23$, thus $\mathbf{E}[N]^3\le\mathbf{E}[N^3]\mathbf{P}(N>0)^2$. Also, the proof of Theorem \ref{thm:thirdmoment} which can be found in Section \ref{sec:th3} shows that $\mathbf{E}[N^3]=\mathbf{E}[N^2]\mathbf{E}_{f}[N^{**}]\approx\mathbf{E}[N]^2\mathbf{E}_{f}[N^{**}]$.
Therefore, $\mathbf{E}_{f}[N^{**}]\ge\mathbf{E}[N]$ holds naturally. Consequently, it suffices to show $\frac{\mathbf{E}_{f}[N^{**}]}{\mathbf{E}[N]}\le1+o(1)$.

As preparations, we state the following propositions which will be important tools in our following analysis.
\begin{proposition}\label{prop:convex}
Assume that $f(x,y,z)$ is a function on $x\in[a_1,a_2],y\in[b_1,b_2],z\in[c_1,c_2]$. If $\frac{\partial^2}{\partial x^2}f\ge0,\frac{\partial^2}{\partial y^2}f\ge0,\frac{\partial^2}{\partial z^2}f\ge0$, and $f(x,y,z)\le0$ on the boundary points. Then for $\forall x\in[a_1,a_2],y\in[b_1,b_2],z\in[c_1,c_2]$,
\begin{align*}
&f(x,y,z)\\
\le&\max\{f(a_1,b_1,c_1),f(a_1,b_1,c_2),f(a_1,b_2,c_1),f(a_1,b_2,c_2),\\
&f(a_2,b_1,c_1),f(a_2,b_1,c_2),f(a_2,b_2,c_1),f(a_2,b_2,c_2)\}\\
\le&0.
\end{align*}
\end{proposition}
\begin{proof}
Due to the condition $\frac{\partial^2}{\partial z^2}f\ge0$, we observe that $f$ is convex with respect to $z$. Since on the boundary points $f(a_1,b_1,c_1)\le0$, $f(a_1,b_1,c_2)\le0$, thus for $\forall z\in[c_1,c_2]$,
$$f(a_1,b_1,z)\le\max\{f(a_1,b_1,c_1),f(a_1,b_1,c_2)\}\le0.$$
Similarly we see that for $\forall z\in[c_1,c_2]$,
 $$f(a_1,b_2,z)\le\max\{f(a_1,b_2,c_1),f(a_1,b_2,c_2)\}\le0.$$

Note that $\frac{\partial^2}{\partial y^2}f\ge0$, thus $f$ is convex with respect to $y$. Therefore similar arguments yield that, for $\forall y\in[b_1,b_2],z\in[c_1,c_2]$,
\begin{align*}
f(a_1,y,z)\le&\max\{f(a_1,b_1,c_1),f(a_1,b_1,c_2),f(a_1,b_2,c_1),f(a_1,b_2,c_2)\}\le0,\\
f(a_2,y,z)\le&\max\{f(a_2,b_1,c_1),f(a_2,b_1,c_2),f(a_2,b_2,c_1),f(a_2,b_2,c_2)\}\le0.
\end{align*}
Finally, the assertion follows analogous arguments since$f$ is also convex with respect to $x$.
\end{proof}

\begin{proposition}\label{prop:multi}
For real numbers $x_1,x_2,...,x_m$ and non-negative integers $n,r_1,r_2,...,r_m,t$ with $t\le n$, then
\begin{align*}
\sum\binom{n}{r_1,r_2,r_3,...,r_m}x_1^{r_1}x_2^{r_2}x_3^{r_3}...x_m^{r_m}=\binom{n}{t}(x_1+x_2)^t(x_3+...+x_m)^{n-t},
\end{align*}
where the sum is over all combinations of $r_1,...,r_m$ such that $r_1+r_2=t$ and $r_1+r_2+...r_m=n$.
\end{proposition}

\begin{proof}
Note that $r_1+r_2=t$ and $r_3+...r_m=n-t$ are fixed, thus
\begin{align*}
&\sum_{r_1,...,r_m}\binom{n}{r_1,r_2,r_3,...,r_m}x_1^{r_1}x_2^{r_2}x_3^{r_3}...x_m^{r_m}\\
=&\sum_{r_1,...,r_m}\frac{n!}{r_1!(t-r_1)!r_3!...r_m!}x_1^{r_1}x_2^{t-r_1}x_3^{r_3}...x_m^{r_m}\\
=&\sum_{r_1,...,r_m}\frac{n!}{t!(n-t)!}\frac{t!}{r_1!(t-r_1)!}\frac{(n-t)!}{r_3!...r_m!}x_1^{r_1}x_2^{t-r_1}x_3^{r_3}...x_m^{r_m}\\
=&\binom{n}{t}\sum_{r_1=0}^t\frac{t!}{r_1!(t-r_1)!}x_1^{r_1}x_2^{t-r_1}\sum_{r_3,...,r_m}\frac{(n-t)!}{r_3!...r_m!}x_3^{r_3}...x_m^{r_m}\\
=&\binom{n}{t}(x_1+x_2)^t(x_3+...+x_m)^{n-t}.
\end{align*}

\end{proof}

\subsection{$\alpha\ge1$}

To evaluate the sum (\ref{eqn:main}), we first observe that

\begin{align*}
\begin{split}
\frac{\partial}{\partial s_1}\psi(s_1,s_2,s)=&-\ln (s_1-s)+\ln(1-s_1-s_2-s_0+2s)
-\ln(d-2)\\
&+\frac{rpks_1^{k-1}\ln n}{1-p+p(s_0^k+s_1^k+s_2^k)+(2p^2-p)s^k/(1-p)}\\
\le&-\ln (s_1-s)+\left(-\alpha+\frac{rpk}{1-p}s_1^{k-1}\right)\ln n.
\end{split}
\end{align*}
It is easy to see that $\frac{\partial}{\partial s_1}\psi(s_1,s_2,s)<0$ for $s_1\in[0,\epsilon]$, where $\epsilon>0$ is a sufficiently small constant. By symmetry of $s_1$ and $s_2$, $\psi(s_1,s_2,s)$ is decreasing with respect to both $s_1$ and $s_2$ around $s_1=s_2=0$. On the other aspect, note that $d=n^\alpha$ and $\alpha\ge1$, thus if $s_1,s_2$ are away from $0$, then $\psi(s_1,s_2,s)$ is dominated by $h(s_1,s_2,s)\ln n$, where
\begin{align}\label{hfunction}
\begin{split}
h(s_1,s_2,s)=-(s_1+s_2-s)\ln d+r\ln n\ln\left[1+\frac{ps_1^k+ps_2^k+(2p^2-p)s^k/(1-p)}{1-p+ps_0^k}\right],
\end{split}
\end{align}
and $0\le s_1,s_2\le 1,0\le s\le\min\{s_1,s_2,s_0\},s_1+s_2-2s\le1-s_0$.

A straightforward calculation yields that
\begin{align*}
\begin{split}
&\frac{\partial^2}{\partial s_1^2}h(s_1,s_2,s)\\
=&rkp\ln n\frac{((1-p)(k-1)-ps_1^k)+p(k-1)(s_0^k+s_2^k)+(2p^2-p)(k-1)s^k/(1-p)}{(1-p+p(s_0^k+s_1^k+s_2^k)+(2p^2-p)s^k/(1-p))^2}s_1^{k-2}.
\end{split}
\end{align*}
Note that $k\ge\frac1{1-p}$ entails that $(1-p)(k-1)-ps_1^k\ge0$, thus $\frac{\partial^2}{\partial s_1^2}h(s_1,s_2,s)\ge0$ on the interval $s_1\in[0,1]$. Similarly, $\frac{\partial^2}{\partial s_2^2}h(s_1,s_2,s)\ge0$ on the interval $s_2\in[0,1]$, and $\frac{\partial^2}{\partial s^2}h(s_1,s_2,s)\ge0$ on the interval $s\in[0,\min\{s_1,s_2,s_0\}]$.

Moreover, on the boundary points we have
\begin{align*}
\begin{split}
h(0,0,0)&=0,\\
h(1,0,0)&=h_1(0,1,0)\le-\ln d+r\ln n\ln\left[1+\frac{p}{1-p}\right]=(-\alpha-r\ln(1-p))\ln n<0,\\
h(1,s_0,s_0)&=h_1(s_0,1,s_0)=(-\alpha-r\ln(1-p))\ln n<0,\\
\end{split}
\end{align*}
therefore the unique maximum point of $h(s_1,s_2,s)$ is $s_1=s_2=s=0$.

 Crucially, we now conclude that the contribution to (\ref{expect1}) comes from the terms near $s_1=0$ and $s_2=0$, say $s_1,s_2\in[0,1/\sqrt n]$. In addition, the sum over $s_1,s_2$ outside the interval $s_1,s_2\in[0,1/\sqrt n]$ is negligible. In fact, by applying Proposition \ref{prop:convex} we know that if $s_1\ge1/\sqrt n$ or $s_2\ge1/\sqrt n$, then $h(s_1,s_2,s)\le-\Theta(\frac1{\sqrt n})$. Consequently,
\begin{align*}
&\sum_{ S_1\ge\sqrt n,(S_1,S_2,S)\in\mathfrak{S}}\Psi(S_1,S_2,S)+\sum_{S_2\ge \sqrt n,(S_1,S_2,S)\in\mathfrak{S}}\Psi(S_1,S_2,S)\\
\le&2\sum_{ S_1\ge\sqrt n,(S_1,S_2,S)\in\mathfrak{S}}\exp\{-\Theta(\sqrt n\ln n)\}\\
=&o(1).
\end{align*}

Finally, for $s_1,s_2\in[0,1/\sqrt n]$, it is straightforward that
\begin{align*}
\left[1+\frac{p(s_1^k+s_2^k)+\frac{2p^2-p}{1-p}s^k}{1-p+ps_0^k}\right]^{rn\ln n}=1+o(1).
\end{align*}
Furthermore, easy calculation yields that
\begin{align}\label{expansion}
\begin{split}
\frac{(d-2)^{n-(S_1+S_2-S)}}{d^n}=\left(1-\frac2d\right)^{S_0-S}\left(\frac1d\right)^{S_1-S}
\left(\frac1d\right)^{S_2-S}\left(\frac1d\right)^{S}\left(1-\frac2d\right)^{n-S_0-S_1-S_2+2S}.
\end{split}
\end{align}

To proceed, note that $(S_0-S)+S=S_0$ is fixed, hence apply Proposition \ref{prop:multi} yields
\begin{align*}
\sum_{0\le S_1,S_2\le \sqrt n,(S_1,S_2,S)\in\mathfrak{S}}&\binom{n}{S_0-S,S_1-S,S_2-S,S,n-S_0-S_1-S_2+2S}
\left(1-\frac2d\right)^{S_0-S}\left(\frac1d\right)^{S}\cdot\\
&\left(\frac1d\right)^{S_1-S}\left(\frac1d\right)^{S_2-S}\left(1-\frac2d\right)^{n-S_0-S_1-S_2+2S}\\
\le&\binom{n}{S_0}\left(1-\frac1d\right)^{S_0}.
\end{align*}
Thus
\begin{align*}\label{expect1}
\begin{split}
&\frac{\mathbf{E}_{f}[N^{**}]}{\mathbf{E}[N]}\le\left(1-\frac1d\right)^{S_0}+o(1)=1+o(1).
\end{split}
\end{align*}
\subsection{$\alpha<1$}

In this section we prove Theorem \ref{expectation} under the condition that $\alpha<1$. The proof of this case is relatively subtle. In fact, for $\alpha\ge1$, $\Psi(S_1,S_2,S)$ attains its global maximum at $s_1=s_2=0$; while if $\alpha<1$, $\psi(S_1,S_2,S)$ is increasing on $S_1-S,S_2-S\in[0,n^{1-\alpha}]$, and then the monotonicity becomes tricky around $n^{1-\alpha}$, which we will deal with delicately later. Due to the fact that the monotonicity involves $S_1-S$ and $S_2-S$, so the use of parameters $S_1,S_2,S$ will lead to unpleasant calculations. To simplify the calculation, sometimes it will be convenient to replace $S_1,S_2$ by the overlap parameters $M_1,M_2$ we have introduced before  (recall that $M_1=S_1-S,M_2=S_2-S$). We will denote
\begin{align*}
\begin{split}
&W(M_1,M_2,S)=\Psi(S_1,S_2,S),\\
&F(m_1,m_2,s)+G(m_1,m_2,s)=\psi(s_1,s_2,s),
\end{split}
\end{align*}
where $M_1=nm_1,M_2=nm_2$, and
\begin{align}
\begin{split}
F(m_1,m_2,s)=&-m_1\ln m_1-m_2\ln m_2-(s_0-s)\ln(s_0-s)-s\ln s+(1-m_1-m_2-s)\ln(d-2)\\
&-(1-m_1-m_2-s_0)\ln(1-m_1-m_2-s_0)-\ln d+s_0\ln s_0+(1-s_0)\ln(1-s_0),\\
G(m_1,m_2,s)=&r\ln n\ln\left[1+\frac{p(m_1+s)^k+p(m_2+s)^k+\frac{2p^2-p}{1-p}s^k}{1-p+ps_0^k}\right].
\end{split}
\end{align}

We begin with the following estimate which is similar to the case of $\alpha\ge1$.
\begin{lemma}\label{EpsilonTo1}
Suppose $\epsilon>0$ is a sufficiently small constant, then
\begin{equation*}
\sum_{S_1,S_2\ge\epsilon n,(S_1,S_2,S)\in\mathfrak{S}}\Psi(S_1,S_2,S)=o(1).
\end{equation*}
\end{lemma}

\begin{proof}
Note that if $s_1,s_2\ge\epsilon$, since $\epsilon>0$ is a constant, then
\begin{align*}
\begin{split}
\psi(s_1,s_2,s)=&-(1+o(1))(s_1+s_2-s)\ln d+r\ln n\ln\left[1+\frac{ps_1^k+ps_2^k+(2p^2-p)s^k/(1-p)}{1-p+ps_0^k}\right].
\end{split}
\end{align*}

Applying Proposition \ref{prop:convex}, and noticing that $s_1\ge\epsilon,s_2\ge\epsilon$, thus we have $h(s_1,s_2,s)\le-\Theta(1)<0$ for any $s_1,s_2\ge\epsilon,0\le s\le\min\{s_0,s_1,s_2\}$. Thus
\begin{equation*}
\sum_{S_1,S_2\ge\epsilon n,(S_1,S_2,S)\in\mathfrak{S}}\Psi(S_1,S_2,S)\\
\le\sum_{S_1,S_2\ge\epsilon n,(S_1,S_2,S)\in\mathfrak{S}}\exp\{-\Theta(n\ln n)\}=o(1).
\end{equation*}

\end{proof}

\begin{lemma}\label{m102d}
\begin{equation*}
\sum_{0\le M_1\le2n/d,(M_1,M_2,S)\in\mathfrak{M}}W(M_1,M_2,S)\le1+o(1).
\end{equation*}
\end{lemma}

\begin{lemma}\label{lemma:3}
\begin{align}
\begin{split}
\sum_{2n/d\le M_1\le\epsilon n,(M_1,M_2,s)\in\mathfrak{M}}W(M_1,M_2,S)=o(1).
\end{split}
\end{align}
\end{lemma}

\begin{lemma}\label{lemma:4}
\begin{align*}
\begin{split}
\sum_{ M_1\ge\epsilon n,(M_1,M_2,s)\in\mathfrak{M}}W(M_1,M_2,S)=o(1).
\end{split}
\end{align*}
\end{lemma}

Consequently, Theorem \ref{expectation} is immediate from Lemmas \ref{EpsilonTo1}, \ref{m102d}, \ref{lemma:3} and \ref{lemma:4}.

\subsubsection{Proof of Lemma 2}

First observe that (\ref{expansion}) can be rewritten as
\begin{align*}
\frac{(d-2)^{n-(S_1+S_2-S)}}{d^n}=\left(1-\frac2d\right)^{S_0-S}\left(\frac1d\right)^{M_1}
\left(\frac1d\right)^{M_2}\left(\frac1d\right)^{S}\left(1-\frac2d\right)^{n-S_0-M_1-M_2}.
\end{align*}

\textbf{Step 1}. $m_2\in[0,2/d]$. We consider $s\in[0,2/d]$ first.
It is easy to see that, for all $0\le M_1,M_2,S\le 2n/d$,
\begin{align*}
G(m_1,m_2,s)\le&G\left(\frac2d,\frac2d,\frac2d\right)\\
=&r\ln n\ln\left[1+\frac{p2^{2k+1}+\frac{2p^2-p}{(1-p)^2}2^k}{1-p+ps_0^k}\frac{1}{d^k}\right]^{rn\ln n}\\
=&\Theta\left(\frac{\ln n}{n^{k\alpha-1}}\right)=o(1),
\end{align*}
since $k\alpha>1$ is a constant. Consequently, $W(M_1,M_2,S)$ is asymptotically equal to the multinomial terms, which can be bounded as follows (note that $(S_0-S)+S=S_0$ is fixed).
\begin{align*}
\sum_{0\le M_1,M_2,S\le 2n/d}&\frac{\binom{n}{S_0-S,M_1,M_2,S,n-S_0-M_1-M_2}}{\binom{n}{S_0}}
\left(1-\frac2d\right)^{S_0-S}\left(\frac1d\right)^{M_1+M_2+S}\left(1-\frac2d\right)^{n-S_0-M_1-M_2}\\
\le&\frac{\binom{n}{S_0}}{\binom{n}{S_0}}\left(1-\frac1d\right)^{S_0}\\
=&\left(1-\frac1d\right)^{S_0}.
\end{align*}

Therefore,
\begin{align}\label{11s2d}
\sum_{0\le M_1,M_2,S\le 2n/d}W(M_1,M_2,S)\le\left(1-\frac1d\right)^{S_0}\le1+o(1).
\end{align}

Thus, if $s_0\le 2/d$, then it is obvious that
\begin{align}\label{eqn:lemma1:s1:1}
\sum_{0\le M_1,M_2\le 2n/d,0\le S\le S_0}W(M_1,M_2,S)\le1+o(1).
\end{align}

If $s_0>2/d$, then for all $m_1,m_2\in[0,2/d]$, and $s\ge2/d$, asymptotically we find that
\begin{align*}
G(m_1,m_2,s)\le&r\ln n\left[1+\frac{p(\frac2d+s)^k+p(\frac2d+s)^k+\frac{2p^2-p}{(1-p)}s^k}{1-p+ps_0^k}\right]\\
=&(1+o(1))r\ln n\ln\left[1+\frac{ps^k}{1-p}\right].
\end{align*}

For any $S$, applying Proposition \ref{prop:multi} again to sum over $M_1,M_2$, we obtain

\begin{align*}
\sum_{0\le M_1,M_2\le 2n/d}&\frac{\binom{n}{S_0-S,M_1,M_2,S,n-S_0-M_1-M_2}}{\binom{n}{S_0}}
\left(1-\frac2d\right)^{S_0-S}\left(\frac1d\right)^{M_1+M_2+S}\left(1-\frac2d\right)^{n-S_0-M_1-M_2}\\
\le&\frac{\binom{n}{S_0-S,S,n-S_0}}{\binom{n}{S_0}}\left(\frac1d\right)^{S}\left(1-\frac1d\right)^{S_0-S}\\
=&\binom{S_0}{S}\left(\frac1d\right)^{S}\left(1-\frac1d\right)^{S_0-S}\\
\le&\binom{n}{S}\left(\frac1d\right)^{S}\left(1-\frac1d\right)^{n-S}.
\end{align*}
Hence,  by an asymptotic calculation we claim that
\begin{align}\label{11ss0}
\begin{split}
&\sum_{S=2n/d}^{S_0}\sum_{0\le M_1,M_2\le 2n/d}W(M_1,M_2,S)\\
\le&\sum_{S=2n/d}^{n}\binom{n}{S}\left(\frac1d\right)^{S}\left(1-\frac1d\right)^{n-S}=o(1),
\end{split}
\end{align}

Combining (\ref{eqn:lemma1:s1:1}), (\ref{11s2d}) and (\ref{11ss0}) yields
\begin{align}\label{11m1m2}
\begin{split}
\sum_{0\le M_1,M_2\le 2n/d}\sum_{S=0}^{S_0}W(M_1,M_2,S)\le1+o(1).
\end{split}
\end{align}

\textbf{Step 2.} $m_2\ge2/d$.

 If $s\in[0,2/d]$, recall that $d^k=n^{k\alpha}$ and $k\alpha>1$ is a constant, thus  it is easy to see that
\begin{align*}
G(m_1,m_2,s)=&rn\ln n\ln\left[1+\frac{p((m_1+s)^k+(m_2+s)^k+\frac{2p^2-p}{(1-p)}s^k}{1-p+ps_0^k}\right]\\
\le&(1+o(1))r\ln n\ln\left[1+\frac{pm_2^k}{1-p}\right].
\end{align*}

Meanwhile, applying Proposition \ref{prop:multi} to sum over $M_1,S$ (note that $(S_0-S)+S$, $M_2$ are fixed), we see that
\begin{align*}
\sum_{0\le M_1,S\le 2n/d}&\frac{\binom{n}{S_0-S,M_1,M_2,S,n-S_0-M_1-M_2}}{\binom{n}{S_0}}
\left(1-\frac2d\right)^{S_0-S}\left(\frac1d\right)^{M_1+M_2+S}\left(1-\frac2d\right)^{n-S_0-M_1-M_2}\\
\le&\frac{\binom{n}{S_0}\binom{n-S_0}{M_2}}{\binom{n}{S_0}}\left(\frac1d\right)^{M_2}\left(1-\frac1d\right)^{n-S_0-M_2}\\
\le&\left(\frac1d\right)^{M_2}\left(1-\frac1d\right)^{n-S_0-M_2}.\\
\end{align*}
Thus
\begin{align}\label{12s2d}
\begin{split}
&\sum_{M_1=0}^{ 2n/d}\sum_{M_2=2n/d}^{n-S_0-M_1}\sum_{S=0}^{2n/d}W(M_1,M_2,S)\\
\le&\sum_{M_2=\frac2dn}^{n-S_0}\binom{n-S_0}{M_2}\left(\frac1d\right)^{M_2}\left(1-\frac1d\right)^{n-S_0-M_2}\\
=&o(1).
\end{split}
\end{align}

If $s_0>2/d$, then for all $s\in[2/d,s_0],m_1\in[0,2/d],m_2\ge2/d$, we derive the following estimate that,
\begin{align*}
G(m_1,m_2,s)=&r\ln n\ln\left[1+\frac{p(m_1+s)^k+(m_2+s)^k+\frac{2p^2-p}{(1-p)}s^k}{1-p+ps_0^k}\right]\\
=&(1+o(1))r\ln n\ln\left[1+\frac{p(m_2+s)^k+\frac{p^2}{1-p}s^k}{1-p+ps_0^k}\right]\\
=&(1+o(1))G_2(m_2,s),
\end{align*}
where $G_2(m_2,s)=r\ln n\ln\left[1+\frac{p(m_2+s)^k+\frac{p^2}{1-p}s^k}{1-p+ps_0^k}\right]$.
At the same time,
\begin{align*}
\sum_{0\le M_1\le2n/d}&\frac{\binom{n}{S_0-S,M_1,M_2,S,n-S_0-M_1-M_2}}{\binom{n}{S_0}}
\left(1-\frac2d\right)^{S_0-S}\left(\frac1d\right)^{M_1+M_2+S}\left(1-\frac2d\right)^{n-S_0-M_1-M_2}\\
\le&\frac{\binom{n}{S_0-S,S,M_2}}{\binom{n}{S_0}}\left(\frac1d\right)^{M_2+S}\left(1-\frac2d\right)^{S_0-S}\left(1-\frac1d\right)^{n-S_0-M_2}\\
=&\exp\{nF_2(m_2,s)\},
\end{align*}
where
\begin{align*}
F_2(m_2,s)=&-m_2\ln m_2-s\ln s-(s_0-s)\ln(s_0-s)-(1-s_0-m_2)\ln(1-s_0-m_2)+s_0\ln s_0\\
&+(1-s_0)\ln(1-s_0)-(m_2+s)\ln d+(s_0-s)\ln(1-\frac2d)+(1-s_0-m_2)\ln(1-\frac1d).
\end{align*}

Now we are ready to simplify the sum as
\begin{align}\label{m2ge2d}
\begin{split}
&\sum_{0\le m_1\le2/d,m_2\ge2/d,s\ge2/d,(m_1,m_2,s)\in\mathfrak{M}}W(M_1,M_2,S)\\
\le&\sum_{m_2\ge2/d,2/d\le s\le s_0}\exp\{nF_2(m_2,s)+nG_2(m_2,s)\}.
\end{split}
\end{align}

To estimate (\ref{m2ge2d}), we consider $m_2,s\in[2/d,\epsilon]$ first. The partial derivatives of $F_2(m_2,s),G_2(m_2,s)$ work out to be
\begin{align}
\begin{split}
&\frac{\partial}{\partial m_2}F_2(m_2,s)=-\ln m_2+\ln(1-m_2-s_0)-\ln d-\ln(1-\frac1d),\\
&\frac{\partial}{\partial s}F_2(m_2,s)=-\ln s+\ln(s_0-s)-\ln d-\ln(1-\frac2d),\\
&\frac{\partial}{\partial m_2}G_2(m_2,s)=rpk\ln n\frac{(m_2+s)^{k-1}}{1-p+ps_0^k+p(m_2+s)^k+p^2s^k/(1-p)},\\
&\frac{\partial}{\partial s}G_2(m_2,s)=rpk\ln n\frac{(m_2+s)^{k-1}+ps^{k-1}/(1-p)}{1-p+p(s_0^k+(m_2+s)^k)+p^2s^k/(1-p)}.
\end{split}
\end{align}

If $m_2,s\in[2/d,1/\sqrt d]$, then $\frac{\partial}{\partial s}(F_2(m_2,s)+G_2(m_2,s))\le-\ln2+\Theta\left(\frac{\ln n}{(\sqrt d)^{k-1}}\right)<0$, since $(\sqrt d)^{-(k-1)}=n^{-(k-1)\alpha/2}$ and $(k-1)\alpha>0$ is a constant.
If $m_2\in[2/d,1/\sqrt d],s\in[1/\sqrt d,\epsilon]$, then $\frac{\partial}{\partial s}(F_2(m_2,s)+G_2(m_2,s))\le-\frac12\ln d+\frac{rpk}{(1-p)^2}\epsilon^{(k-1)}\ln n=\big(-\frac12\alpha+\frac{rpk}{(1-p)^2}\epsilon^{(k-1)}\big)\ln n<0$, since $\epsilon>0$ is a sufficiently small constant.
Therefore, if $m_2\in[2/d,1/\sqrt d]$ then $F_2(m_2,s)+G_2(m_2,s)$ is decreasing on $s\in[2/d,\epsilon]$. Similar arguments yield the monotonicity of  $F_2(m_2,s)+G_2(m_2,s)$ as shown in Table \ref{tab:2}.

\begin{table}[H]
\caption{Monotonicity of $F_2(m_2,s)+G_2(m_2,s)$}
\centering
\label{tab:2}
\normalsize
\begin{tabular}{cccc}
  \hline
  &  & $\quad\frac{\partial}{\partial m_2}(F_2(m_2,s)+G_2(m_2,s))$& $\quad\frac{\partial}{\partial s}(F_2(m_2,s)+G_2(m_2,s))$ \\
  \hline
  $m_2\in[2/d,1/\sqrt d],$&$s\in[2/d,1/\sqrt d]$&\quad $<0$  &\quad $< 0$ \\
  $m_2\in[2/d,1/\sqrt d],$&$s\in[1/\sqrt d,\epsilon]$&unknown  &\quad $<0$ \\
 $m_2\in[1/\sqrt d,\epsilon],$&$s\in[2/d,1/\sqrt d]$&\quad $<0$ &unknown \\
  $m_2\in[1/\sqrt d,\epsilon],$&$s\in[1/\sqrt d,\epsilon]$ &\quad $<0$&\quad $<0$  \\
  \hline
\end{tabular}
\end{table}

Now given the monotonicity of $F_2(m_2,s)+G_2(m_2,s)$ we see that, if $m_2\in[2/d,1/\sqrt d],s\in[2/d,\epsilon]$, then
\begin{align*}
&F_2(m_2,s)+G_2(m_2,s)\\
\le &F_2\left(m_2,\frac2d\right)+G_2(m_2,\frac2d)\le F_2\left(\frac2d,\frac2d\right)+G_2\left(\frac1{\sqrt d},\frac2d\right)\\
=&-\frac4d\ln\frac2d-\left(s_0-\frac2d\right)\ln\left(s_0-\frac2d\right)+s_0\ln s_0-\left(1-s_0-\frac2d\right)\ln\left(1-s_0-\frac2d\right)+(1-s_0)\ln(1-s_0)\\
&-\frac4d\ln d+\left(s_0-\frac2d\right)\ln\left(1-\frac2d\right)-\left(1-s_0-\frac2d\right)\ln\left(1-\frac1d\right)\\
&+r\ln n\ln\left[1+\frac{p(\frac2d+\frac1{\sqrt d})^k+\frac{p^2}{1-p}(\frac1{\sqrt d})^k}{1-p}\right].\\
\end{align*}

As a first step to estimate this term, we claim  that
\begin{align}\label{lns0}
\nonumber&-\left(s_0-\frac2d\right)\ln\left(s_0-\frac2d\right)+s_0\ln s_0-\left(1-s_0-\frac2d\right)\ln\left(1-s_0-\frac2d\right)+(1-s_0)\ln(1-s_0)\\
\le&(1+o(1))\frac4d.
\end{align}

In fact, (1) if $s_0=o(1)$, then $-(s_0-\frac2d)\ln(s_0-\frac2d)+s_0\ln s_0<0$ since $s\ln s$ is decreasing on $(0, 1/e)$. At the same time $-(1-s_0-\frac2d)\ln(1-s_0-\frac2d)+(1-s_0)\ln(1-s_0)\le-(1-s_0)[\ln(1-s_0-\frac2d)-\ln(1-s_0)]=-(1-s_0)\ln(1-\frac2{d(1-s_0)})=(1+o(1))\frac2d$.

(2) If $s_0,1-s_0=O(1)$, $s_0\ln s_0-(s_0-\frac2d)\ln(s_0-\frac2d)=\frac2d\ln(s_0-\frac2d)-s_0\ln(1-\frac2{ds_0})\le-s_0\ln(1-\frac2{ds_0})=(1+o(1))\frac2d$. And $(1-s_0)\ln(1-s_0)-(1-s_0-\frac2d)\ln(1-s_0-\frac2d)\le(1-s_0)\ln(1-\frac4{d(1-s_0)})=(1+o(1))\frac2d$.

(3) If $s_0=1-o(1)$, $s_0\ln s_0-(s_0-\frac2d)\ln(s_0-\frac2d)=
(1+o(1))\frac2d$. And $(1-s_0)\ln(1-s_0)-(1-s_0-\frac2d)\ln(1-s_0-\frac2d)<0$.

Combining the above yields (\ref{lns0}) as desired.

To proceed, we also note that $(s_0-\frac2d)\ln(1-\frac2d)+(1-s_0-\frac4d)\ln(1-\frac1d)\le-\frac1d$ since $\frac1d=o(1)$, thus
\begin{align}\label{F2G21}
F_2(m_2,s)&+G_2(m_2,s)\le-\frac4d\ln2+\frac3d+\Theta\left(\frac{\ln n}{(\sqrt d)^k}\right).
\end{align}

Analogously, if $m_2\in[1/\sqrt d,\epsilon],s\in[2/d,1/\sqrt d]$, we have
\begin{align}\label{F2G22}
\begin{split}
F_2(m_2,s)&+G_2(m_2,s)\le F_2\left(\frac1{\sqrt d},s\right)+G_2\left(\frac1{\sqrt d},s\right)\le F_2\left(\frac1{\sqrt d},\frac2d\right)+G_2\left(\frac1{\sqrt d},\frac1{\sqrt d}\right)\\
\le&-(1+o(1))\frac{1}{2\sqrt d}\ln d+\Theta\left(\frac{\ln n}{(\sqrt d)^k}\right),\\
\end{split}
\end{align}
and for $m_2\in[1/\sqrt d,\epsilon],s\in[1/\sqrt d,\epsilon]$ we have
\begin{align}\label{F2G23}
\begin{split}
F_2(m_2,s)&+G_2(m_2,s)\le F_2\left(\frac1{\sqrt d},\frac1{\sqrt d}\right)+G_2\left(\frac1{\sqrt d},\frac1{\sqrt d}\right)\\
\le&-(1+o(1))\frac{1}{\sqrt d}\ln d+\Theta\left(\frac{\ln n}{(\sqrt d)^k}\right),\\
\end{split}
\end{align}

In summary, it follows from (\ref{F2G21}), (\ref{F2G22}) and (\ref{F2G23}) as desired that
\begin{align}\label{eqn:F2first}
\begin{split}
\sum_{m_2,s=2/d}^{\epsilon}\exp\{nF_2(m_2,s)+nG_2(m_2,s)\}=o(1).
\end{split}
\end{align}

Secondly, if $m_2\ge\epsilon$ or $s\ge\epsilon$, then $F_2(m_2,s)=-(1+o(1))(m_2+s)\ln d$, and $$F_2(m_2,s)+G_2(m_2,s)=-(1+o(1))s_2\ln d+r\ln n\ln\left[1+\frac{ps_2^k+\frac{p^2}{1-p}s^k}{1-p+ps_0^k}\right],$$ where $s_2=m_2+s$.

We claim that
\begin{align}\label{hs2s}
  g(s_2,s) &= r\ln\left[1+\frac{ps_2^k+\frac{p^2}{1-p}s^k}{1-p+ps_0^k}\right]-\alpha s_2\le0,
\end{align}
where $0\le s_2\le1,0\le s\le\min\{s_2,s_0\}.$

Indeed, recalling the conditions that $k\ge\frac1{1-p}$, $s_2\le1$ and $s\le s_0,s_2$, we can obtain
\begin{align*}
  \frac{\partial^2}{\partial s_2^2}g(s_2,s)=rpk\frac{(k-1)(1-p)-ps_2^k+(k-1)(ps_0^k+\frac{p^2}{1-p}s^k)}{(1-p+ps_0^k+ps_2^k+\frac{p^2}{1-p}s^k)^2}s_2^{k-2}&>0,\\
\frac{\partial^2}{\partial s^2}g(s_2,s)=\frac{rkp^2}{1-p}\frac{(k-1)(1-p)+(k-1)p(s_0^k+s_2^k)-\frac{p^2}{1-p}s^k}{(1-p+ps_0^k+ps_2^k+\frac{p^2}{1-p}s^k)^2}s^{k-2}&>0.
\end{align*}
Further, on the boundary points
\begin{align*}
g(0,0)&=0,\\
g(1,0)&=r\ln\left[1+\frac{p}{1-p+ps_0^k}\right]-\alpha\le-r\ln(1-p)-\alpha<0,\\
g(1,s_0)&=-r\ln(1-p)-\alpha<0,
\end{align*}
where the inequality holds since $r\le r_c=\frac{\alpha}{-\ln(1-p)}$. Thus it follows from Proposition \ref{prop:convex} that $g(s_2,s)=-\Theta(1)<0$ if we require that $m_2>\epsilon$ or $s\ge\epsilon$. Consequently, the contribution to the sum in this case is bounded by
\begin{align}\label{eqn:F2second}
\begin{split}
\sum_{m_2\ge2/d,s\ge\epsilon}\exp\{nF_2(m_2,s)+nG_2(m_2,s)\}
+\sum_{m_2\ge\epsilon,s\ge2/d}\exp\{nF_2(m_2,s)+nG_2(m_2,s)\}=o(1).
\end{split}
\end{align}
It follows from (\ref{eqn:F2first}) and (\ref{eqn:F2second}) that
\begin{align}\label{eqn:F2}
\begin{split}
\sum_{m_2\ge2/d,s\ge2/d}\exp\{nF_2(m_2,s)+nG_2(m_2,s)\}=o(1).
\end{split}
\end{align}

The lemma follows from (\ref{11m1m2}), (\ref{12s2d}),(\ref{m2ge2d}) and (\ref{eqn:F2}).


\subsubsection{Proof of Lemma 3}

Note that in this section $m_1\in[2/d,\epsilon]$.
The partial derivatives of $F(m_1,m_2,s)$ and $G(m_1,m_2,s)$ with respect to $m_1,s$ work out to be
\begin{equation}\label{partialFG}
\begin{split}
&\frac{\partial}{\partial m_1}F(m_1,m_2,s)=-\ln m_1+\ln(1-m_1-m_2-s_0)-\ln(d-2),\\
&\frac{\partial}{\partial s}F(m_1,m_2,s)=\ln (s_0-s)-\ln s-\ln(d-2),\\
&\frac{\partial}{\partial m_1}G(m_1,m_2,s)=\frac{rpk(m_1+s)^{k-1}\ln n}{1-p+p(s_0^k+(m_1+s)^k+(m_2+s)^k)+(2p^2-p)s^k/(1-p)},\\
&\frac{\partial}{\partial s}G(m_1,m_2,s)=rkp\ln n\frac{(m_1+s)^{k-1}+(m_2+s)^{k-1}+(2p-1)s^{k-1}/(1-p)}{1-p+p(s_0^k+(m_1+s)^k+(m_2+s)^k)+(2p^2-p)s^k/(1-p)}.
\end{split}
\end{equation}
We see that $F(m_1,m_2,s)$ is decreasing on $m_1\ge1/d.$ By symmetry of $m_1$ and $m_2$, $F(m_1,m_2,s)$ is also decreasing on $m_1\ge1/d$.

\textbf{Step 1.} Assume that $m_2\in[0,2/d]$. Then by Step 2 of the proof of Lemma \ref{m102d} together with the symmetric property that $W(M_1,M_2,S)=W(M_2,M_1,S)$, we have

\begin{align}\label{eqn:s1}
\begin{split}
\sum_{M_1= 2n/d}^{n\epsilon}\sum_{M_2=0}^{2n/d}\sum_{S=0}^{S_0}W(M_1,M_2,S)=o(1).
\end{split}
\end{align}

\textbf{Step 2.} Assume that $m_2\in[2/d,\epsilon]$.

\begin{align}
\begin{split}
\sum_{M_1= 2n/d}^{n\epsilon}\sum_{M_2= 2n/d}^{n\epsilon}\sum_{S=0}^{S_0}W(M_1,M_2,S)=o(1).
\end{split}
\end{align}

 For any  $0\le s\le 2/d$, note that $d^k=n^{k\alpha}$ and $k\alpha>1$ is a constant, thus  it is easy to see that
\begin{align*}
G(m_1,m_2,s)=&r\ln n\ln\left[1+\frac{p((m_1+s)^k+(m_2+s)^k+\frac{2p^2-p}{(1-p)}s^k}{1-p}\right]\\
\le&(1+o(1))r\ln n\ln\left[1+\frac{pm_1^k+pm_2^k}{1-p}\right].
\end{align*}

At the same time,

\begin{align*}
&\sum_{0\le S\le 2n/d}\frac{\binom{n}{S_0-S,M_1,M_2,S,n-S_0-M_1-M_2}}{\binom{n}{S_0}}
\left(1-\frac2d\right)^{S_0-S}\left(\frac1d\right)^{M_1}\left(\frac1d\right)^{M_2}\left(\frac1d\right)^{S}\left(1-\frac2d\right)^{n-S_0-M_1-M_2}\\
\le&\frac{\binom{n}{S_0}\binom{n-S_0}{M_1,M_2}}{\binom{n}{S_0}}\left(\frac1d\right)^{M_1+M_2}\left(1-\frac1d\right)^{n-S_0-M_1-M_2}\\
=&\frac{(n-S_0)!}{M_1!M_2!(n-S_0-M_1-M_2)!}\left(\frac1d\right)^{M_1+M_2}\left(1-\frac1d\right)^{n-S_0-M_1-M_2}\\
=&\exp\Big\{n\Big[(1-s_0)\ln(1-s_0)-m_1\ln m_1-m_2\ln m_2-(1-s_0-m_1-m_2)\ln(1-s_0-m_1-m_2)\\
&-(m_1+m_2)\ln d+(1-s_0-m_1-m_2)\ln\Big(1-\frac1d\Big)\Big]\Big\}.
\end{align*}

Denote
\begin{align*}
f(m_1,m_2)=&(1-s_0)\ln(1-s_0)-m_1\ln m_1-m_2\ln m_2-(1-s_0-m_1-m_2)\ln(1-s_0-m_1-m_2)\\
&-(m_1+m_2)\ln d+(1-s_0-m_1-m_2)\ln(1-\frac1d)+r\ln n\ln\left[1+\frac{pm_1^k+pm_2^k}{1-p}\right].
\end{align*}
Then we have
\begin{align*}
\begin{split}
\sum_{M_1= 2n/d}^{n\epsilon}\sum_{M_2= 2n/d}^{n\epsilon}\sum_{S=0}^{2n/d}W(M_1,M_2,S)\le\sum_{M_1= 2n/d}^{n\epsilon}\sum_{M_2= 2n/d}^{n\epsilon}\exp\{nf(m_1,m_2)\}.
\end{split}
\end{align*}
Note that
\begin{align*}
\frac{\partial}{\partial m_1}f(m_1,m_2)=&-\ln m_1+\ln(1-s_0-m_1-m_2)-\ln d+
r\ln n\frac{pkm_1^{k-1}}{1-p+pm_1^k+pm_2^k}\\
\le&-\ln m_1-\ln d+
\frac{rpk}{1-p}m_1^{k-1}\ln n.
\end{align*}
If $m_1\in[2/d,1/\sqrt d]$, then $\frac{\partial}{\partial m_1}f(m_1,m_2)\le-\ln2+\frac{rpk}{1-p}n^{-(k-1)\alpha}\ln n<0$, since $(k-1)\alpha>0$ is a constant. If $m_1\in[1/\sqrt d,\epsilon]$, then $\frac{\partial}{\partial m_1}f(m_1,m_2)\le-\frac12\ln d+\frac{rpk}{1-p}\epsilon^{(k-1)\ln n}=\big(-\frac12\alpha+\frac{rpk}{1-p}\epsilon^{(k-1)}\big)\ln n<0$, since $\epsilon>0$ can be sufficiently small.
Therefore, $f(m_1,m_2)$ is decreasing on $m_1\in[2/d,\epsilon]$, and by symmetry  $f(m_1,m_2)$ is also decreasing on $m_2\in[2/d,\epsilon]$.
Thus
\begin{align*}
f(m_1,m_2)\le f\left(\frac2d,\frac2d\right)=&(1-s_0)\ln(1-s_0)-\frac4d\ln \frac2d-\left(1-s_0-\frac4d\right)\ln\left(1-s_0-\frac4d\right)\\
&-\frac4d\ln d+\left(1-s_0-\frac4d\right)\ln\left(1-\frac1d\right)+r\ln n\ln\left[1+\frac{2p(\frac2d)^k}{1-p}\right].\\
\end{align*}
A similar argument as (\ref{lns0}) implies that
$$(1-s_0)\ln(1-s_0)-\left(1-s_0-\frac4d\right)\ln\left(1-s_0-\frac4d\right)\le(1+o(1))\frac4d.$$
Consequently,
\begin{align*}
f(m_1,m_2)\le -\frac4d\ln2+\frac4d+\frac{2rp2^k}{1-p}\frac{\ln n}{d^k}+o\left(\frac1d\right)=-\frac2d\left(2\ln2-2-\frac{rp2^k}{1-p}\frac{\ln n}{d^{k-1}}+o(1)\right).\\
\end{align*}
Therefore
\begin{align}\label{eqn:s2:1}
\begin{split}
&\sum_{M_1= 2n/d}^{n\epsilon}\sum_{M_2= 2n/d}^{n\epsilon}\sum_{S=0}^{2n/d}W(M_1,M_2,S)\\
\le&\sum_{M_1= 2n/d}^{n\epsilon}\sum_{M_2= 2n/d}^{n\epsilon}\sum_{S=0}^{2n/d}\exp\{nf(m_1,m_2)\}\\
\le&\sum_{M_1= 2n/d}^{n\epsilon}\sum_{M_2= 2n/d}^{n\epsilon}\sum_{S=0}^{2n/d}\exp\left\{-2n^{1-\alpha}\Big(2\ln2-2-\frac{rp2^k}{1-p}\frac{\ln n}{d^{k-1}}\Big)\right\}=o(1).
\end{split}
\end{align}

If $s_0\le2/d$, it is obvious that
\begin{align}\label{eqn:s2:2}
\begin{split}
\sum_{M_1= 2n/d}^{n\epsilon}\sum_{M_2= 2n/d}^{n\epsilon}\sum_{S=0}^{S_0}W(M_1,M_2,S)\le o(1).
\end{split}
\end{align}

If $s_0\ge2/d$, then we need to consider the sum over $s\in[2/d,s_0],m_1,m_2\in[2/d,\epsilon]$, which will be worked out by considering the values of $m_1,m_2,s$ delicately.

	\textbf{Case 1.}  $m_1\in[2/d,1/\sqrt d]$. We will consider the following six sub-cases depending on the values of $m_2$ and $s$.

\textbf{Case 1.1. } $m_2,s\in[2/d,1/\sqrt d]$. Note that
\begin{align*}
\frac{\partial}{\partial m_1}(F(m_1,m_2,s)+G(m_1,m_2,s))\le&-\ln \frac2d-\ln(d-2)+rpk\ln n\left(\frac2{\sqrt d}\right)^{k-1}/(1-p)\\
=-\ln2&-\ln(1-\frac2d)+\frac{rpk}{1-p}\ln n\left(\frac2{\sqrt d}\right)^{k-1}=-\ln2+o(1)<0,\\
\frac{\partial}{\partial s}(F(m_1,m_2,s)+
G(m_1,m_2,s))\le &-\ln\frac2d-\ln(d-2)\\
&+ rkp\ln n\frac{(\frac2{\sqrt d})^{k-1}+(\frac2{\sqrt d})^{k-1}+(2p-1)(\frac1{\sqrt d})^{k-1}/(1-p)}{1-p}\\
=&-\ln2+o(1)<0.
\end{align*}

Thus in the case of $m_1,m_2,s\in[2/d,1/\sqrt d]$, $F(m_1,m_2,s)+ G(m_1,m_2,s)$ is decreasing with respect to $m_1,m_2$ and $s$, hence
\begin{align*}
F(m_1,m_2,s)+ G(m_1,m_2,s)\le F\left(\frac2d,\frac2d,\frac2d\right)+G\left(\frac2d,\frac2d,\frac2d\right).
\end{align*}
To be precise,
\begin{align*}
F\left(\frac2d,\frac2d,\frac2d\right)=&-\frac4d\ln\frac2d-\left(s_0-\frac2d\right)\ln\left(s_0-\frac2d\right)-\frac2d\ln\frac2d-
\left(1-\frac4d-s_0\right)\ln\left(1-\frac4d-s_0\right)\\
&+\left(1-\frac6d\right)\ln(d-2)
-\ln d+s_0\ln s_0+(1-s_0)\ln(1-s_0)\\
=&-\frac6d\ln2+\frac6d\ln d+\left(1-\frac6d\right)\ln d+\left(1-\frac6d\right)\ln\left(1-\frac2d\right)-\ln d
-\left(s_0-\frac2d\right)\ln\left(s_0-\frac2d\right)\\
&+s_0\ln s_0+\frac4d\ln\left(1-\frac4d-s_0\right)+(1-s_0)\left(\ln(1-s_0)-\ln\big(1-\frac4d-s_0\big)\right).
\end{align*}
A similar calculation as in (\ref{lns0}) shows that
$$-\left(s_0-\frac2d\right)\ln\left(s_0-\frac2d\right)+s_0\ln s_0+(1-s_0)\left(\ln(1-s_0)-\ln\big(1-\frac4d-s_0\big)\right)\le\frac6d+o\left(\frac1d\right).$$
Moreover, $(1-\frac6d)\ln(1-\frac2d)=-\frac2d+o(\frac1d)$. Then

$$F\left(\frac2d,\frac2d,\frac2d\right)\le-\frac6d\ln2+\frac4d+o\left(\frac1d\right).$$

At the same time,
\begin{align*}
 G\left(\frac2d,\frac2d,\frac2d\right)=&r\ln n\ln\left[1+\frac{p((\frac4d)^k+(\frac4d)^k)+(2p^2-p)(\frac2d)^k/(1-p)}{1-p+ps_0^k}\right]\\
\le&r\ln n\ln\left[1+\frac{2p4^k+(2p^2-p)2^k/(1-p)}{1-p}\frac1{d^k}\right]\\
=&\Theta\left(\frac{\ln n}{n^{k\alpha}}\right).
\end{align*}
Therefore
\begin{align}
\begin{split}
\sum_{M_1,M_2,S=\frac2dn}^{n/\sqrt d}W(M_1,M_2,S)\le&\sum_{M_1,M_2,S=\frac2dn}^{n/\sqrt d}\exp\left\{nF(\frac2d,\frac2d,\frac2d)+ nG(\frac2d,\frac2d,\frac2d)\right\}\\
\le&\sum_{M_1,M_2,S=\frac2dn}^{n/\sqrt d}\exp\left\{n\left(-\frac6d\ln2+\frac4d+\Theta\left(\frac{\ln n}{n^{k\alpha-1}}\right)\right)\right\}\\
=&o(1).
\end{split}
\end{align}
\textbf{Case 1.2. } $m_2\in[2/d,1/\sqrt d],s\in[1/\sqrt d,\epsilon]$, where $\epsilon>0$ is a sufficiently small constant. Then
\begin{align*}
&\frac{\partial}{\partial s}(F(m_1,m_2,s)+G(m_1,m_2,s))\\
\le&-\ln\frac1{\sqrt d}-\ln(d-2)+kpr\ln n\frac{2(2/d+\epsilon)^{k-1}+(2p-1)\epsilon^{k-1}/(1-p)}{1-p}\\
=&-\frac{\alpha}2\ln n+\frac{kpr}{(1-p)^2} \epsilon^{k-1}\ln n.
\end{align*}
Since $\epsilon>0$ can be sufficiently small, then $F(m_1,m_2,s)+G(m_1,m_2,s)$ is decreasing on $s\in[1/\sqrt d,\epsilon]$ in this case.
\begin{align*}
F(m_1,m_2,s)+r\ln nG(m_1,m_2,s)\le F\left(\frac2d,\frac2d,\frac1{\sqrt d}\right)+G\left(\frac1{\sqrt d},\frac1{\sqrt d},\frac1{\sqrt d}\right). \\
\end{align*}
Specifically, we estimate $F(\frac2d,\frac2d,\frac1{\sqrt d})$ first.
\begin{align*}
F\left(\frac2d,\frac2d,\frac1{\sqrt d}\right)=&-\frac4d\ln\frac2d-\left(s_0-\frac1{\sqrt d}\right)\ln\left(s_0-\frac1{\sqrt d}\right)-\frac1{\sqrt d}\ln\frac1{\sqrt d}-\left(1-\frac4d-s_0\right)\ln\left(1-\frac4d-s_0\right)\\
&+\left(1-\frac4d-\frac1{\sqrt d}\right)\ln(d-2)-\ln d+s_0\ln s_0+(1-s_0)\ln(1-s_0)\\
=&-\frac{4\ln2}{d}-\frac1{2\sqrt d}\ln d+\left(1-\frac4d-\frac1{\sqrt d}\right)\ln\left(1-\frac2d\right)
+s_0\ln s_0\\
&-\left(s_0-\frac1{\sqrt d}\right)\ln\left(s_0-\frac1{\sqrt d}\right)+(1-s_0)\ln(1-s_0)-\left(1-\frac4d-s_0\right)\ln\left(1-\frac4d-s_0\right).
\end{align*}
By a similar argument with before, we obtain
\begin{align*}
F\left(\frac2d,\frac2d,\frac1{\sqrt d}\right)\le&-(1+o(1))\frac{1}{2\sqrt d}\ln d.
\end{align*}
Additionally,
\begin{align}\label{sqrtGd}
\begin{split}
 G(\frac1{\sqrt d},\frac1{\sqrt d},\frac1{\sqrt d})=&r\ln n\ln\left[1+\frac{p(\frac2{\sqrt d})^k+p(\frac2{\sqrt d})^k+\frac{2p^2-p}{1-p}(\frac1{\sqrt d})^k}{1-p+ps_0^k}\right]\\
\le&r\ln n\ln\left[1+\frac{2p2^k+\frac{2p^2-p}{1-p}}{1-p}\frac1{(\sqrt d)^k}\right]\\
=&\Theta\left(\frac{\ln n}{(\sqrt d)^k}\right).
\end{split}
\end{align}

Thus
\begin{align}
\begin{split}
\sum_{M_1,M_2=\frac2dn}^{n/\sqrt d}\sum_{S=n/\sqrt d}^{\epsilon n}W(M_1,M_2,S)\le&\sum_{M_1,M_2=\frac2dn}^{n/\sqrt d}\sum_{S=n/\sqrt d}^{\epsilon n}
\exp\left\{nF\left(\frac2d,\frac2d,\frac1{\sqrt d}\right)+ nG\left(\frac1{\sqrt d},\frac1{\sqrt d},\frac1{\sqrt d}\right)\right\}\\
=&\sum_{M_1,M_2=\frac2dn}^{n/\sqrt d}\sum_{S=n/\sqrt d}^{\epsilon n}\exp\left\{n\left(-(1+o(1))\frac{1}{2\sqrt d}\ln d+\Theta\left(\frac{\ln n}{(\sqrt d)^k}\right)\right)\right\}\\
=&o(1).
\end{split}
\end{align}

\textbf{Case 1.3. }  $m_2\in[2/d,1/\sqrt d],s\in[\epsilon,s_0]$.
Note that in this case  $s_1=m_1+s\ge\epsilon,s_2=m_2+s\ge\epsilon$, thus it follows immediately from Lemma \ref{EpsilonTo1} that
\begin{align}
\begin{split}
\sum_{M_1,M_2=\frac2dn}^{n/\sqrt d}\sum_{S=\epsilon n}^{S_0}W(M_1,M_2,S)=o(1).
\end{split}
\end{align}

\textbf{Case 1.4. } $ m_2\in[1/\sqrt d,\epsilon],s\in[2/d,1/\sqrt d]$.
By symmetry of $m_1,m_2$ and from (\ref{partialFG}),
\begin{align}
\begin{split}
&\frac{\partial}{\partial m_2}(F(m_1,m_2,s)+G(m_1,m_2,s))\\
\le&-\ln \frac1{\sqrt d}-\ln(d-2)+(1+o(1))\frac{rpk}{1-p}\epsilon^k\ln n\\
=&-\left(\frac12\alpha-(1+o(1))\frac{rpk}{1-p}\epsilon^k\right)\ln n,
\end{split}
\end{align}
which is negative since $\epsilon>0$ is sufficiently small. Thus $F(m_1,m_2,s)+G(m_1,m_2,s)$ is decreasing on $ m_2\in[1/\sqrt d,\epsilon]$ in this case. Also, (\ref{partialFG}) entails that $F(m_1,m_2,s)$ is decreasing on $m_1,s\in[2/d,1/\sqrt d]$, while $G(m_1,m_2,s)$ is increasing on $m_1,s\in[2/d,1/\sqrt d]$, hence
\begin{align*}
\begin{split}
&F(m_1,m_2,s)+G(m_1,m_2,s)\\
\le &F\left(m_1,\frac1{\sqrt d},s\right)+G\left(m_1,\frac1{\sqrt d},s\right)\\
\le &F\left(\frac2d,\frac1{\sqrt d},\frac2d\right)+G\left(\frac1{\sqrt d},\frac1{\sqrt d},\frac1{\sqrt d}\right)\\
=&-\frac4d\ln\frac2d-\frac1{\sqrt d}\ln\frac1{\sqrt d}-\left(s_0-\frac1{\sqrt d}\right)\ln\left(s_0-\frac1{\sqrt d}\right)-\left(1-\frac2d-\frac1{\sqrt d}-s_0\right)\ln\left(1-\frac2d-\frac1{\sqrt d}-s_0\right)\\
&+\left(1-\frac4d-\frac1{\sqrt d}\right)\ln(d-2)
-\ln d+s_0\ln s_0+(1-s_0)\ln(1-s_0)+G\left(\frac1{\sqrt d},\frac1{\sqrt d},\frac1{\sqrt d}\right)\\
=&-\frac4d\ln 2-\frac1{2\sqrt d}\ln d+\left(1-\frac4d-\frac1{\sqrt d}\right)\ln\left(1-\frac2d\right)-\left(s_0-\frac1{\sqrt d}\right)\ln\left(s_0-\frac1{\sqrt d}\right)\\
&+s_0\ln s_0-\left(1-\frac2d-\frac1{\sqrt d}-s_0\right)\ln\left(1-\frac2d-\frac1{\sqrt d}-s_0\right)+(1-s_0)\ln(1-s_0)\\
\le&-\frac1{2\sqrt d}\ln d+(1+o(1))\frac2{\sqrt d}+\Theta\left(\frac{\ln n}{(\sqrt d)^k}\right).
\end{split}
\end{align*}
Combining with (\ref{sqrtGd}) yields the desired result that
\begin{align}
\begin{split}
\sum_{M_1=\frac2dn}^{n/\sqrt d}\sum_{M_2=n/\sqrt d}^{\epsilon n}\sum_{S=\epsilon n}^{S_0}W(M_1,M_2,S)=o(1).
\end{split}
\end{align}

\textbf{Case 1.5. } $m_2,s\in[1/\sqrt d,\epsilon]$. From (\ref{partialFG}),
\begin{align*}
\begin{split}
\frac{\partial}{\partial m_2}F(m_1,m_2,s)+\frac{\partial}{\partial m_2}G(m_1,m_2,s)&\le-\ln \frac1{\sqrt d}-\ln(d-2)+r\ln n\frac{kp2^{k-1}}{1-p}\epsilon^k\\
\frac{\partial}{\partial s}F(m_1,m_2,s)+\frac{\partial}{\partial s}G(m_1,m_2,s)&\le-\ln \frac1{\sqrt d}-\ln(d-2)+rkp\ln n\frac{2^{k-1}+p/(1-p)}{1-p}\epsilon^k\\
=&-\left(\frac12\alpha-\frac{rpk(2^{k-1}+p/(1-p))}{1-p}\epsilon^k\right)\ln n-\ln\left(1-\frac1d\right),
\end{split}
\end{align*}
where both partial derivatives are negative since $r,k,p$ are constants and $\epsilon>0$ is sufficiently small.
Therefore
\begin{align}
\begin{split}
&F(m_1,m_2,s)+G(m_1,m_2,s)\le F\left(\frac2d,\frac1{\sqrt d},\frac1{\sqrt d}\right)+G\left(\frac1{\sqrt d},\frac1{\sqrt d},\frac1{\sqrt d}\right)
\end{split}
\end{align}
Specifically,
\begin{align*}
\begin{split}
 &F\left(\frac2d,\frac1{\sqrt d},\frac1{\sqrt d}\right)\\
 =&-\frac2d\ln\frac2d-\frac2{\sqrt d}\ln\frac1{\sqrt d}-\left(1-\frac2d-\frac2{\sqrt d}\right)\ln(d-2)-\ln d+s_0\ln s_0-\left(s_0-\frac1{\sqrt d}\right)\ln\left(s_0-\frac1{\sqrt d}\right)\\
&-\left(1-\frac2d-\frac1{\sqrt d}-s_0\right)\ln\left(1-\frac2d-\frac1{\sqrt d}-s_0\right)+(1-s_0)\ln(1-s_0).
\end{split}
\end{align*}
A similar argument as in (\ref{lns0}) implies that
\begin{align*}
&s_0\ln s_0-\left(s_0-\frac1{\sqrt d}\right)\ln\left(s_0-\frac1{\sqrt d}\right)+(1-s_0)\ln(1-s_0)\\
&-\left(1-\frac2d-\frac1{\sqrt d}-s_0\right)\ln\left(1-\frac2d-\frac1{\sqrt d}-s_0\right)\\
\le&(1+o(1))\frac2{\sqrt d}.
\end{align*} 
Thus,
\begin{align*}
 F\left(\frac2d,\frac1{\sqrt d},\frac1{\sqrt d}\right)\le&-(1+o(1))\frac{\ln d}{\sqrt d}.
\end{align*}
Further, recall from  (\ref{sqrtGd}) that
\begin{align*}
G\left(\frac1{\sqrt d},\frac1{\sqrt d},\frac1{\sqrt d}\right)=&\Theta\left(\frac{\ln n}{(\sqrt d)^k}\right).
\end{align*}
It is now immediate that
\begin{align}
\begin{split}
&\sum_{M_1=2n/d}^{n/\sqrt d}\sum_{M_2=n/\sqrt d}^{\epsilon n}\sum_{S=n/\sqrt d}^{\epsilon n}W(M_1,M_2,S)\\
\le&\sum_{M_1=2n/d}^{n/\sqrt d}\sum_{M_2=n/\sqrt d}^{\epsilon n}\sum_{S=n/\sqrt d}^{\epsilon n}\exp\left\{nF(m_1,m_2,s)+ nG(m_1,m_2,s)\right\}\\
\le&\sum_{M_1=2n/d}^{n/\sqrt d}\sum_{M_2=n/\sqrt d}^{\epsilon n}\sum_{S=n/\sqrt d}^{\epsilon n}\exp\left\{-(1+o(1))\frac{n\ln d}{\sqrt d}+\Theta\left(\frac{n\ln n}{(\sqrt d)^k}\right)\right\}\\
=&o(1).
\end{split}
\end{align}

\textbf{Case 1.6. } $m_2\in[1/\sqrt d,\epsilon],s\in[\epsilon,s_0]$. Recalling that $m_1\in[2/d,1/\sqrt d]$, thus we can claim that
 $F(m_1,m_2,s)$ is dominated by $(-m_2-s)\ln d$. Hence we have
\begin{align*}
\begin{split}
&F(m_1,m_2,s)+G(m_1,m_2,s)\\
=&(1+o(1))\left[-\alpha (m_2+s)+r\ln\left(1+\frac{p(m_2+s)^k+p^2s^k/(1-p)}{1-p+ps_0^k}\right)\right]\ln n\\
=&(1+o(1))g(m_2,s)\ln n,
\end{split}
\end{align*}
where $g(m_2,s)$ is defined by (\ref{hs2s}). By a similar argument, $g(m_2,s)=-\Theta(1)<0$ if $m_2\in[1/\sqrt d,\epsilon],s\in[\epsilon,s_0]$.

\begin{align}
\begin{split}
&\sum_{M_1=2n/d}^{n/\sqrt d}\sum_{M_2=n/\sqrt d}^{\epsilon n}\sum_{S=\epsilon n}^{S_0}W(M_1,M_2,S)\\
\le&\sum_{M_1=2n/d}^{n/\sqrt d}\sum_{M_2=n/\sqrt d}^{\epsilon n}\sum_{S=\epsilon n}^{S_0}\exp\{-\Theta(n\ln n)\}\\
=&o(1).
\end{split}
\end{align}

Combining cases 1.1-1.6, we obtain that
\begin{align}\label{eqn:s2:case1}
\begin{split}
\sum_{M_1=2n/d}^{n/\sqrt d}\sum_{M_2=2n/d}^{\epsilon n}\sum_{S=2n/d}^{S_0}\exp\{n\Psi(m_1,m_2,s)\}=o(1).
\end{split}
\end{align}

\textbf{Case 2. } $m_1\in[1/\sqrt d,\epsilon]$.
By symmetry of $m_1$ and $m_2$, and from results of cases 1.4-1.6 we see that
\begin{align}\label{eqn:case2:0}
\begin{split}
\sum_{M_1=n/\sqrt d}^{\epsilon n}\sum_{M_2=2n/d}^{ n/\sqrt d}\sum_{S=2n/d}^{S_0}\exp\{n\Psi(m_1,m_2,s)\}=o(1).
\end{split}
\end{align}

\textbf{Case 2.1. }  $m_2\in[1/\sqrt d,\epsilon]$, $s\in[2/d,1/\sqrt d]$.
Note that
\begin{align*}
\begin{split}
\frac{\partial}{\partial m_1}(F(m_1,m_2,s)+G(m_1,m_2,s))&\le-\ln \frac1{\sqrt d}-\ln(d-2)+\frac{rpk}{1-p}\left(\epsilon+\frac1{\sqrt d}\right)^{k-1}\ln n\\
&=-\frac12\alpha\ln n+\frac{rpk}{1-p}\left(\epsilon+\frac1{\sqrt d}\right)^{k-1}\ln n+o(1),
\end{split}
\end{align*}
which is negative since $\epsilon>0$ can be sufficiently small. Moreover, by symmetry of $m_1,m_2$ and note that $\frac{\partial}{\partial s}F(m_1,m_2,s)<0$ on $s\in[2/d,\sqrt d]$, then

\begin{align*}
\begin{split}
F(m_1,m_2,s)+G(m_1,m_2,s)\le F\left(\frac1{\sqrt d},\frac1{\sqrt d},\frac2d\right)+G\left(\frac1{\sqrt d},\frac1{\sqrt d},\frac1{\sqrt d}\right).
\end{split}
\end{align*}
Firstly, note that
\begin{align*}
\begin{split}
 &F\left(\frac1{\sqrt d},\frac1{\sqrt d},\frac2d\right)\\
 =&-\frac2{\sqrt d}\ln\frac1{\sqrt d}-\frac2d\ln\frac2d+\left(1-\frac2{\sqrt d}-\frac2d\right)\ln(d-2)-\ln d-\left(s_0-\frac2d\right)\ln\left(s_0-\frac2d\right)\\
 &-\left(1-\frac2{\sqrt d}-s_0\right)\ln\left(1-\frac2{\sqrt d}-s_0\right)+s_0\ln s_0+(1-s_0)\ln(1-s_0)\\
 =&-(1+o(1))\frac1{\sqrt d}\ln d.
\end{split}
\end{align*}
At the same time,  (\ref{sqrtGd}) entails that
\begin{align*}
\begin{split}
G(\frac1{\sqrt d},\frac1{\sqrt d},\frac1{\sqrt d})=\Theta\left(\frac{\ln n}{(\sqrt d)^k}\right).
\end{split}
\end{align*}

Therefore,
\begin{align}\label{eqn:case2:1}
\begin{split}
&\sum_{M_1=n/\sqrt d}^{\epsilon n}\sum_{M_2=n/\sqrt d}^{\epsilon n}\sum_{S=2n/d}^{n/\sqrt d}\exp\{nF(m_1,m_2,s)+nG(m_1,m_2,s)\}\\
\le&\sum_{M_1=n/\sqrt d}^{\epsilon n}\sum_{M_2=n/\sqrt d}^{\epsilon n}\sum_{S=2n/d}^{n/\sqrt d}\exp\left\{n\ln n(-(1+o(1))\frac\alpha{\sqrt d}+\Theta\left(n^{-k\alpha/2}\right)\right\}\\
=&o(1).
\end{split}
\end{align}

\textbf{Case 2.2. } $m_2\in[1/\sqrt d,\epsilon]$, $s\in[1/\sqrt d,\epsilon]$.
In this case, (\ref{partialFG}) entails that $F(m_1,m_2,s)+G(m_1,m_2,s)$ decreases with respect to $m_1,m_2,s$. Thus
\begin{align*}
\begin{split}
F(m_1,m_2,s)+G(m_1,m_2,s)\le F\left(\frac1{\sqrt d},\frac1{\sqrt d},\frac1{\sqrt d}\right)+G\left(\frac1{\sqrt d},\frac1{\sqrt d},\frac1{\sqrt d}\right).
\end{split}
\end{align*}
We can estimate it as before that
\begin{align*}
\begin{split}
 F\left(\frac1{\sqrt d},\frac1{\sqrt d},\frac1{\sqrt d}\right)=&-\frac3{\sqrt d}\ln\frac1{\sqrt d}+\left(1-\frac3{\sqrt d}\right)\ln(d-2)-\ln d-\left(s_0-\frac1{\sqrt d}\right)\ln\left(s_0-\frac1{\sqrt d}\right)\\
 &-\left(1-\frac2{\sqrt d}-s_0\right)\ln\left(1-\frac2{\sqrt d}-s_0\right)+s_0\ln s_0+(1-s_0)\ln(1-s_0)\\
 =&-(1+o(1))\frac3{2\sqrt d}\ln d.
\end{split}
\end{align*}
At the same time
\begin{align*}
\begin{split}
G(\frac1{\sqrt d},\frac1{\sqrt d},\frac1{\sqrt d})=\Theta\left(\frac{\ln n}{(\sqrt d)^k}\right).
\end{split}
\end{align*}
Therefore,
\begin{align}\label{eqn:case2:2}
\begin{split}
&\sum_{M_1=n/\sqrt d}^{\epsilon n}\sum_{M_2=n/\sqrt d}^{\epsilon n}\sum_{S==n/\sqrt d}^{\epsilon n}\\
=&\sum_{M_1=n/\sqrt d}^{\epsilon n}\sum_{M_2=n/\sqrt d}^{\epsilon n}\sum_{S==n/\sqrt d}^{\epsilon n}\exp\left\{n\ln n\left(-(1+o(1))\frac{3\alpha}{2}n^{-\alpha/2}+\Theta\left(n^{-k\alpha/2}\right)\right)\right\}\\
=&o(1).
\end{split}
\end{align}

\textbf{Case 2.3. } $m_2\in[1/\sqrt d,\epsilon]$, $s\in[\epsilon,s_0]$. In this case we can see that $s_1=m_1+s\ge\epsilon,s_2=m_2+s\ge\epsilon$.
It follows by applying Lemma \ref{EpsilonTo1} that
\begin{align}\label{eqn:case2:3}
\begin{split}
\sum_{M_1=n/\sqrt d}^{\epsilon n}\sum_{M_2=n/\sqrt d}^{\epsilon n}\sum_{S=\epsilon n}^{S_0}W(M_1,M_2,S)=o(1).
\end{split}
\end{align}

It follows from (\ref{eqn:case2:0}), (\ref{eqn:case2:1}),(\ref{eqn:case2:2}) and (\ref{eqn:case2:3}) that
\begin{align}\label{eqn:s2:case2}
\begin{split}
\sum_{M_1=n/\sqrt d}^{\epsilon n}\sum_{M_2=2n/d}^{\epsilon n}\sum_{S=2n/d}^{S_0}W(M_1,M_2,S)=o(1).
\end{split}
\end{align}

Combining (\ref{eqn:s1}), (\ref{eqn:s2:1}), (\ref{eqn:s2:2}), (\ref{eqn:s2:case1}) and (\ref{eqn:s2:case2}) completes the proof.

\subsubsection{Proof of Lemma 4.}

Assume that $m_2\in[0,2/d]$. By symmetry of $m_1$ and $m_2$ and Step 2 of the proof of Lemma \ref{m102d}, we have
\begin{align}\label{geepsilon1}
\begin{split}
\sum_{M_1\ge\epsilon n,0\le M_2\le 2n/d,(M_1,M_2,s)\in\mathfrak{M}}W(M_1,M_2,S)=o(1).
\end{split}
\end{align}

If $m_2\in[2/d,\epsilon]$. By symmetry of $m_1$ and $m_2$ and Step 2 and 3 of the proof of Lemma \ref{lemma:3}, we have

\begin{align}\label{geepsilon2}
\begin{split}
\sum_{M_1\ge\epsilon n,2n/d\le M_2\le \epsilon n,(M_1,M_2,s)\in\mathfrak{M}}W(M_1,M_2,S)=o(1).
\end{split}
\end{align}

Finally, if $m_2\ge\epsilon$. Now note that $s_1=m_1+s\ge\epsilon,s_2=m_2+s\ge\epsilon$, then recall the results of Lemma \ref{EpsilonTo1}, we have
\begin{align}\label{geepsilon3}
\begin{split}
\sum_{M_1\ge\epsilon n, M_2\ge \epsilon n,(M_1,M_2,s)\in\mathfrak{M}}W(M_1,M_2,S)=o(1).
\end{split}
\end{align}
Combining (\ref{geepsilon1}) ,(\ref{geepsilon2}) and (\ref{geepsilon3}) completes the proof.
\section{Proof of Theorem \ref{thm:hard} }\label{proof:th2}

We prove Corollaries \ref{distribution} and \ref{thm:thirdmoment} first.
\subsection{
Proof of Corollary \ref{distribution}}
First note that the conditional probability series $\{\mathcal{P}[\sigma \text{ is a solution of }I|\mathcal{E}_t],t\ge1\}$ is an increasing sequence with respect to $t$, thus
\begin{align*}
 \mathcal{P}[\sigma \text{ is a solution of }I|\mathcal{E}_t]\ge\mathcal{P}[\sigma \text{ is a solution of }I|\mathcal{E}_1]=\mathcal{P} [\sigma\text{ is a solution of }I].
\end{align*}
Further, for any instance $I\in\mathbb{S}^*$ which contains a hidden assignment $\sigma$, using Bayes' rule we see that
\begin{align*}
   \mathcal{P}_t^*&=\mathbf{Pr}[\mathcal{E}_t|\sigma\text{ is a solution }] \\
  &=\frac{\mathbf{Pr} [\sigma\text{ is a solution of }I,\mathcal{E}_t]}{\mathbf{Pr} [\sigma\text{ is a solution of }I]} \\
  &=\frac{\mathbf{Pr} [\sigma\text{ is a solution of }I|\mathcal{E}_t]}{\mathbf{Pr} [\sigma\text{ is a solution of }I]}\mathbf{Pr}(\mathcal{E}_t)\\
&\ge\mathbf{Pr}(\mathcal{E}_t)=\mathcal{P}_t.
\end{align*}
Moreover, the expected number of satisfying assignments of instances that are forced to satisfy $\sigma$ is
\begin{align*}
  \mathbf{E}_f[N^*] &=\sum_{t=1}^{d^n}t\cdot \mathbf{Pr}[\mathcal{E}_t|\sigma\text{ is a solution }]=\sum_{t=1}^{d^n}t\mathcal{P}_t^*.
\end{align*}
Note that $\mathbf{E}[N]=\sum_{t=1}^{d^n}t\mathcal{P}_t$, also according to \cite{xu2000} it holds that
\begin{align*}
 \lim_{n\rightarrow\infty} \frac{\mathbf{E}_f[N^*]}{\mathbf{E}[N]}= \lim_{n\rightarrow\infty} \frac{\mathbf{E}[N^2]}{\mathbf{E}^2[N]}=1.
\end{align*}
In addition, combining the fact that $\mathcal{P}_t^*\ge\mathcal{P}_t,$ we have
\begin{align*}
\mathcal{P}_t^*\sim\mathcal{P}_t.
\end{align*}

Analogously, for any instance $I\in\mathbb{S}^{**}$ which contains two satisfying assignments, say $\sigma,\tau$, the probability that $I$ contains $t$ solutions is $ \mathcal{P}_t^{**}=\mathbf{Pr}[\mathcal{E}_t|\sigma,\tau\text{ are solutions of }I] $. A similar analysis yields that

\begin{align*}
   \mathcal{P}_t^{**}
  &=\frac{\mathbf{Pr} [\sigma,\tau\text{ are solutions of }I,\mathcal{E}_t]}{\mathbf{Pr} [\sigma,\tau\text{ are solutions of }I]} \\
  &=\frac{\mathbf{Pr} [\sigma,\tau\text{ are solutions of }I|\mathcal{E}_t]}{\mathbf{Pr} [\sigma,\tau\text{ are solutions of }I]}\mathbf{Pr}(\mathcal{E}_t)\\
&\ge\mathbf{Pr}(\mathcal{E}_t)=\mathcal{P}_t.
\end{align*}
Furthermore, the expected number of solutions for
\begin{align*}
  \mathbf{E}_f[N^{**}] &=\sum_{t=1}^{d^n}t\cdot \mathbf{Pr}[\mathcal{E}_t|\sigma,\tau\text{ are solutions }]=\sum_{t=1}^{d^n}t\mathcal{P}_t^{**}.
\end{align*}
Combining $\mathcal{P}_t^{**}\ge\mathcal{P}_t$ and Theorem \ref{expectation} that,
\begin{align*}
 \lim_{n\rightarrow\infty} \frac{\mathbf{E}_f[N^{**}]}{\mathbf{E}[N]}=1,
\end{align*}
we obtain
\begin{align*}
 \mathcal{P}_t^{**}\sim\mathcal{P}_t.
\end{align*}


\subsection{Proof of Corollary \ref{thm:thirdmoment}}\label{sec:th3}
Writing the third moment out as a sum over all possible triplets, we obtain
\begin{align*}
& \mathbf{E}[N^3]=\sum_{\sigma,\tau,\omega}\mathbf{Pr}[\sigma,\tau,\omega\text{ are satisfying assignments}]\\
&=\sum_{\sigma,\tau,\omega}\mathbf{Pr}[\omega\text{ is a satisfying assignment }|\sigma,\tau\text{ are satisfying assignments}]\mathbf{Pr}[\sigma,\tau\text{ are satisfying assignments}]\\
&=\mathbf{E}_{f}[N^{**}]\sum_{\sigma,\tau}\mathbf{Pr}[\sigma,\tau\text{ are satisfying assignments}]\\
&=\mathbf{E}_{f}[N^{**}]\mathbf{E}[N^2].
\end{align*}

Recall the results from \cite{xu2000} that $\lim_{n\rightarrow\infty}\frac{\mathbf{E}[N^2]}{\mathbf{E}[N]^2}=1$, and combine with Theorem \ref{expectation} completes the proof.
\subsection{Proof of Theorem \ref{thm:hard}}
From Theorem \ref{expectation} we know that, the expected number of solutions of forced RB instances with two hidden solutions is asymptotically the same with that of unforced RB instances, which means that the sizes of the solution spaces $\mathbb{S}$ and $\mathbb{S}^{**}$ are asymptotically equal. In addition, Corollaries \ref{distribution} and  \ref{thm:thirdmoment} entail that, the solutions of random RB instances are distributed uniformly, and more importantly the distribution of random forced RB instances with two arbitrary hidden solutions in the entire instance space is asymptotically and uniformly the same with that of random unforced RB instances.  Therefore the important facts that the expected number of solutions and distribution of instances are asymptotically the same determine that, the performance of an algorithm on a randomly generated forced RB instance with two hidden solutions will be same with unforced ones.

\section{Hiding exact maximum independent sets in graphs}

 As shown in \cite{ju1998}, finding any large clique in random graphs with a constant number of hidden cliques is still hard, thus by considering cliques as private keys and the entire graph as a public key, hidden cliques provide an elementary means of creating cryptographically secure primitives. It is well known that an independent set of a graph is also a clique of the complementary graph. Therefore our results suggest applications to cryptography in that, we provide a way to hide \emph{exact} maximum independent sets (MISs) with precise size $n$ in a random graph model generated based on model RB.

To be precise, by considering the variables and binary clauses as vertices and edges respectively, one can generate a random graph model basing on Model RB as follows (http://sites.nlsde.buaa.edu.cn/~kexu/):
 \begin{itemize}
   \item  Generate $n$ disjoint cliques, each contains $n^\alpha$ vertices (where $\alpha>0$ is a constant).
   \item  Randomly select two distinct cliques, and generate $pn^{2\alpha}$ random edges between the two cliques (without repetition, and $0<p<1$ is a constant).
\item Run step 2 for another $rn\ln n-1$ times (with repetition, and $r>0$ is a constant).
\end{itemize}
For such a graph model, the size of the maximum independent set (MIS) is at most $n$. The upper bound $n$ can be reached if and only if there exists a solution of the corresponding RB instances.

To generate a random graph with a hidden independent set of size $n$, we can use the following steps:
 \begin{itemize}
   \item  Generate $n$ disjoint cliques, each contains $n^\alpha$ vertices (where $\alpha>0$ is a constant).
   \item  Select from each clique a vertex at random to form an independent set of size $n$.
   \item  Randomly select two distinct cliques, and generate $pn^{2\alpha}$ random edges between the two cliques (without repetition, and $0<p<1$ is a constant), where no edge is allowed to violate the maximum independent set.
   \item Run step 3 for another $rn\ln n-1$ times (with repetition, and $r>0$ is a constant).
\end{itemize}

In fact, one can hide two or more MISs following the process above.


\section{Conclusions}

In this paper, we have shown that the expected number of solutions, the distribution as well as the hardness of forced RB instances with two arbitrary hidden solutions  are asymptotically the same with the ones of unforced RB instances.

We close with several questions for future work. First, it would be interesting to investigate the algorithm performance on forced RB instances. Second, determine the maximum number of solutions  we can hide while without violating most of the characteristics of the ensemble, e.g. expected number of solutions, the distribution of the number of solutions, the hardness at the threshold and solution space structure.  Third, higher order moment arguments will be involved if we plant more solutions, and are potentially difficult to evaluate. Does there exist relatively feasible mathematical techniques to tackle such problems?

\vspace{0.3cm}
\textbf{Acknowledgment.} I am very grateful to Prof. Ke Xu for suggesting the problem, and for helpful discussions.


\end{document}